\documentclass{article}

\usepackage{amsmath,amssymb,amsthm}

\usepackage[round]{natbib}
\usepackage{fourier}
\usepackage[T1]{fontenc}
\usepackage{pgfplots}
\pgfplotsset{compat=1.18}

\usepackage{hyperref}
\usepackage{cleveref}
\usepackage{nicefrac}

\usepackage{complexity}
\usepackage{MnSymbol}
\usepackage[bb=boondox]{mathalfa}

\usepackage[ruled,noend]{algorithm2e}

\usepackage{thm-restate}
\usepackage{fullpage}
\usepackage{authblk}

\newtheorem{theorem}{Theorem}[section]

\newtheorem{proposition}[theorem]{Proposition}
\newtheorem{lemma}[theorem]{Lemma}

\theoremstyle{definition}

\newtheorem{example}[theorem]{Example}

\usepackage{tikz}
    \usetikzlibrary{positioning,arrows,shapes,calc}
    \usetikzlibrary{decorations.pathreplacing}
    \usetikzlibrary{patterns}

\tikzset{
            voter/.style={circle,draw,minimum size=0.5cm,inner sep=0, fill = black!05, font=\footnotesize},
          }

\title{Selecting Interlacing Committees}

\author[1]{Chris Dong}
\author[2]{Martin Bullinger}
\author[2]{Tomasz W\k{a}s}
\author[3]{Larry Birnbaum}
\author[3]{Edith Elkind}

\affil[1]{\small School of Computation, Information and Technology, Technical University of Munich, Germany}
\affil[2]{ \small Department of Computer Science, University of Oxford, UK}
\affil[3]{ \small School of Engineering, Northwestern University, USA\protect\\ \vspace*{0.1cm} chris.dong@tum.de, \{martin.bullinger,tomasz.was\}@cs.ox.ac.uk, \{l-birnbaum,edith.elkind\}@northwestern.edu}

\date{}

\newcommand{\Pairs}{\textsc{Pairs}}

\newcommand{\Connections}{\textsc{Cons}}

\newcommand{\AV}{\textsc{AV}}
\newcommand{\CC}{\textsc{CC}}

\newcommand{\EJR}{\textsc{EJR}}

\newcommand{\allpairs}{V^{(2)}}

\newcommand{\orderof}{\mathcal O}

\newcommand{\score}{\mathbf{score}}
\newcommand{\concomp}[2]{K_{#2}(#1)}
\newcommand{\conpair}{\Connections}
\newcommand{\dyn}{\mathbf{dp}}
\newcommand{\opt}{\mathbf{opt}}

\definecolor{myred}{HTML}{db3f3d}

\begin{document}

\maketitle

\begin{abstract}
    Polarization is a major concern for a well-functioning society. Often, mass polarization of a society is driven by polarizing political representation, even when the latter is easily preventable. The existing computational social choice methods for the task of committee selection are not designed to address this issue. We enrich the standard approach to committee selection by defining two quantitative measures that evaluate how well a given committee interconnects the voters. Maximizing these measures aims at avoiding polarizing committees. While the corresponding maximization problems are \NP-complete in general, we obtain efficient algorithms for profiles in the voter-candidate interval domain. Moreover, we analyze the compatibility of our goals with other representation objectives, such as excellence, diversity, and proportionality. We identify trade-offs between approximation guarantees, and describe algorithms that achieve simultaneous constant-factor approximations.
\end{abstract}


\section{Introduction}
In recent years, the increasing prevalence of polarization has been a global concern, discussed not just by social scientists, but by society at large, and accompanied by extensive media coverage~\citep{FAP11a,LeMa16a}.
Polarization is commonly defined as the division of a group into clusters of completely different opinions or ideologies \citep{FiAb08a}.
It may result in
greater ideological extremes and a reduced willingness to compromise or engage with differing views.
As such, polarization is a major roadblock for the modern society, which has to work towards a consensus when resolving global challenges, such as fighting poverty, climate change, or pandemics (see \cite{levin2021dynamics} and the references therein).

Importantly, polarization can occur as a phenomenon concerning an entire society or only at the level of political representation, e.g., when considering the distribution of opinions among the delegates in a parliament.
The former is often referred to as \emph{mass polarization}, while the latter is known as \emph{elite polarization} \citep[see, e.g.,][]{AbSa08a}.

Academic literature broadly agrees that the phenomenon of elite polarization is on the rise.
For example, when depicting the members of the US Congress in terms of their ideology on a scale ranging from the most liberal to the most conservative, one can observe a significant shift when comparing the 87th Congress in the 1960s and the 111th Congress around 2010, see Figure 2.1 in the book by \citet{Fior17a}.
However, whether the society as a whole is polarized as well is less clear.
\citet{FAP11a} argue that there is no conclusive evidence for mass polarization, even when considering highly sensitive topics such as abortion.
For instance, they provide evidence that the elite polarization among delegates is already much higher than the polarization among party identifiers \citep[][Table~2.1]{FAP11a}.
They argue that the media play an important role in creating an inaccurate picture of mass polarization~\citep{FAP11a}.
Indeed, the media can have a significant effect on the perception of and conclusions drawn from elite polarization~\citep{LeMa16a}.

This view is opposed by \citet{AbSa08a}, who analyze data from the American National Election Studies.
They provide extensive evidence that mass polarization has increased significantly since the 1970s.
Moreover, their results suggest mass polarization based on geography (i.e., different ideologies across US states) or religious beliefs.

Against this background, we aim to offer a novel perspective on the intertwined phenomena of mass polarization at the broad level of a society as a whole and elite polarization at the level of the society's political, parliamentary representation.
We highlight how an election can lead to a parliament that is far more polarized than the society it represents, and we propose quantitative measures that evaluate a set of representatives according to how well it interlaces the electorate.
We believe that our ideas can be developed to prevent
societies with broadly moderate opinions from being represented by
unnecessarily polarized parliaments.

We approach polarized democratic representation through the lens of social choice theory.
In this line of research, parliamentary elections have been conceptualized as so-called multiwinner voting rules.
Their formal study, especially in the approval-based setting, in which each voter's ballot specifies a set of approved candidates, has received extensive attention in recent years \citep{FSST17a,LaSk22b}.

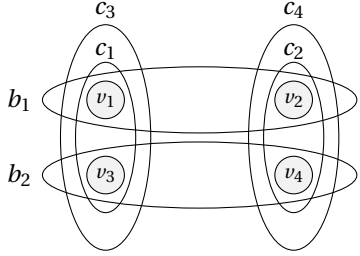
\begin{figure}
    \centering

    \begin{tikzpicture}
        \node[voter] (v1) at (0,1) {$v_1$};
        \node[voter] (v2) at (2.5,1) {$v_2$};
        \node[voter] (v3) at (0,0) {$v_3$};
        \node[voter] (v4) at (2.5,0) {$v_4$};

            \draw[scale = 1.1] ($(v1)!.5!(v2)$) ellipse (1.9cm and .4cm);
            \draw[scale = 1.1] ($(v3)!.5!(v4)$) ellipse (1.9cm and .4cm);
            \node at ($(v1)+(-1.15,0)$) {$b_1$};
            \node at ($(v3)+(-1.15,0)$) {$b_2$};

            \draw[scale = 1] ($(v1)!.5!(v3)$) ellipse (.4cm and 1cm);
            \draw[scale = 1.5] ($(v1)!.5!(v3)$) ellipse (.4cm and 1cm);
            \node at ($(v1)+(0,.65)$) {$c_1$};
            \node at ($(v1)+(0,1.2)$) {$c_3$};
            \draw[scale = 1] ($(v2)!.5!(v4)$) ellipse (.4cm and 1cm);
            \draw[scale = 1.5] ($(v2)!.5!(v4)$) ellipse (.4cm and 1cm);
            \node at ($(v2)+(0,.65)$) {$c_2$};
            \node at ($(v2)+(0,1.2)$) {$c_4$};
    \end{tikzpicture}

    \caption{A preference profile with four voters $v_1,\dots, v_4$ is depicted as hypergraph, where the voters are nodes and the candidates $b_i, c_j$ are hyperedges connecting the voters approving them. In this profile, typical multiwinner voting rules
    do not distinguish between selecting $\{c_1,c_2,c_3,c_4\}$ and $\{c_1,c_2,b_1,b_2\}$.}
    \label{fig:motivation1}
    \end{figure}

\begin{example}
\label{ex:motivation}
As a motivating example, consider the voting scenario illustrated in \Cref{fig:motivation1}.
There are four voters, indicated by the gray circles, as well as six candidates.
Each candidate is represented by an ellipse that covers the voters approving this candidate.
For instance, candidate~$b_1$ is approved by voters~$v_1$ and~$v_2$, whereas candidates~$c_1$ and~$c_3$ are both approved by the same set of voters, namely~$v_1$ and~$v_3$. In practice, this is likely to happen when $c_1$ and $c_3$ represent very similar ideologies.

Assume that we want to select a committee consisting of $4$ candidates.
Two reasonable choices would be to select $W = \{c_1,c_2,c_3,c_4\}$ or $W' = \{c_1,c_2,b_1,b_2\}$.
Both selections lead to committees in which each voter approves exactly two selected candidates.
Moreover, multiwinner voting rules typically considered in the literature, such as Thiele rules and their sequential variants \citep{Thie95a}, Phragm\'{e}n's rule \citep{Phra99a}, or the more recently introduced method of equal shares \citep{PeSk20a}, do not distinguish between these two choices.
There is, however, a difference.
While $W$ divides the electorate into two perfectly separated subsets of voters, $W'$ connects all voters.
From the perspective of polarization, $W$ looks polarizing while $W'$ bridges all voters.
Thus, we need novel voting rules that can tease out this distinction.
In our paper, we aim to provide a principled approach that favors committees in the spirit of $W'$.\footnote{Of course, while we try to highlight the phenomenon at hand with a simple example, our construction extends to elections with many voters or candidates: e.g., each voter in the example might represent a quarter of a large electorate.}
\hfill$\lhd$
\end{example}

We define two simple objectives that aim to measure how well a committee interlaces the voters.
First, we consider maximizing the number of {\em pairs} of voters approving a common candidate (the \Pairs{} objective).
While optimizing this objective leads to the selection of $W'$ in \Cref{ex:motivation}, it can still result in voters being split into large disconnected clusters (cf.~\Cref{example:connections}).
The reason is that \Pairs{} only counts direct, but not indirect links.
Hence, as a second objective we count the number of pairs of voters that are {\em connected} by a sequence of candidates (the \Connections{} objective).

While both objectives immediately give rise to voting rules---select a committee maximizing \Pairs{} or \Connections{}---we primarily view them as measures of polarization.
Whenever they are high, polarization in the selected committee is low.
Thus, we investigate the feasibility of maximizing our objectives, both on their own and in combination with the goals of diversity and proportionality.

We first consider the computational problem of maximizing \Pairs{} or \Connections{} in isolation
(\Cref{sec:comp}).
Unfortunately, for unrestricted preferences this problem is \NP-hard.
However, we obtain a polynomial-time algorithm for the structured domain of voter-candidate interval (VCI) preferences \citep{GodBatSkoFal-2021-2DApprovals},
where voters and candidates are represented by intervals on the real line and a voter approves a candidate if and only if their intervals intersect.
Such preferences are reasonable in parliamentary elections where candidates can often be ordered on a left-right spectrum and voters approve candidates that are close to them on this spectrum.

In \Cref{sec:trade-offs}, we investigate whether one can select interlacing committees while achieving other desiderata.
We first consider \emph{excellence}, as measured by the \emph{approval voting} (\AV{}) score, i.e., the total number of approvals received by committee members.
There is a straightforward way to obtain what is essentially an $\alpha$-approximation of the \Pairs{} objective together with a $(1 - \alpha)$-approximation of the \AV{} score: one can simply use an $\alpha$-fraction of the committee for the former and a $(1 - \alpha)$-fraction for the latter.
Unfortunately, it turns out that this simple algorithm is essentially optimal:
We prove that if a voting rule provides an $\alpha$-approximation of the \Pairs{} objective and a $\beta$-approximation of the \AV{} score, then necessarily $\alpha + \beta \le 1$.
Next, we look at \emph{diversity}, as captured by the
\emph{Chamberlin--Courant} (\CC{}) score,
which is the number of voters who approve at least one candidate in the committee.
The \CC{} score is closely related to the \Pairs{} objective: the former measures the coverage of voters, while the latter measures the coverage of pairs of voters.
Hence, it is quite surprising that we obtain the same trade-off as for \Pairs{} and \AV{}.
Further, we study the compatibility with \emph{proportionality}, as captured by the extended justified representation axiom (\EJR{}).
Again, we show the same tight trade-off:
If a voting rule provides an $\alpha$-approximation of the \Pairs{} objective and a $\beta$-approximation of EJR, then $\alpha + \beta \le 1$.

It is more challenging to combine the \Connections{} objective with \AV{}, \CC{}, \EJR{}, or even \Pairs{}.
This is due to an interesting qualitative difference between \Pairs{} and \Connections{}.
While a constant fraction of the best candidates achieves a constant approximation of \Pairs{},
for \Connections{} this is not the case.
Hence, we obtain worse trade-offs:
If a voting rule provides an $\alpha^2$-approximation of \Connections{} and a $\beta$-approximation of \AV{}, \CC{}, \EJR{}, or \Pairs{}, then $\alpha + \beta \le 1$.
Note that since $\alpha < 1$, it holds that $\alpha^2 < \alpha$.
Hence, for instance, $\alpha^2 = \frac 13$ and $\beta = \frac 12$ is already impossible.
Moreover, for \Connections{} and \AV{} specifically,
the trade-off that we obtain is even more subtle, which suggests that finding a matching lower bound might be challenging.
Nevertheless, we make first steps towards this goal, by
showing that under suitable domain restrictions there always exists a committee that achieves a $\frac 14$-approximation of \Connections{} and a $\frac 12$-approximation of \AV{}, \CC{}, \EJR{}, or \Pairs{}, which matches our upper bound.

\section{Related Work}

In the existing literature,
multiwinner voting rules usually aim
to guarantee
the selection of the best candidates based on their individual quality
\citep{BarCoe-2008-CommitteeMonotonicity,ElkFalSkoSli-2017-Multiwinner},
representation of diverse opinions
\citep{ChaCou-1983-ChamberlinCourant,ElkIsm-2015-CCforOWA},
or proportional treatment of cohesive voter groups
\citep{Thie95a,Phra99a,Mon-1995-Monroe,PeSk20a}.
An overview of the most common
approval-based multiwinner voting rules
is given in the book by~\citet{LaSk22b}.
To the best of our knowledge,
no rules were proposed so far
with the explicit goal of
reducing polarization or
connecting voters.

A line of research in multiwinner voting looks at the possibility
of combining various objectives
as well as their inherent trade-offs,
similar to our study in \Cref{sec:trade-offs}.
\citet{LacSko-2020-MultiwinnerApproximations} provide worst-case bounds on \AV{} and \CC{} scores of committees output by popular voting rules.
For ordinal preferences, \citet{KocKolElkFalTal-2019-MultigoalMultiwinner} analyze the complexity of finding committees that offer an optimal combination of approximations of two objectives.
Moreover, a series of works look at \AV{} and \CC{} scores that can be guaranteed by committees that satisfy proportionality axioms \citep{BriPet-2024-MultigoalMultiwinner,ElkFalIgaManSchEtal-2024-PriceOfJR,FaiVilMeiGal-2022-MultigoalPB}.

A number of authors study the relationship
between an electoral system (or, more narrowly, a voting rule) and the way the candidates
choose to strategically place themselves on the political spectrum
\citep{Cox-1985-EquilibriumApproval,MyeWeb-1993-VotingEquilibria,BolMatTroXef-2019-VotingEquilibria,KurBar-2024-PolarizingProp}.
Such an analysis can indicate whether a rule
prevents, or reinforces, polarization.
Our approach differs in that we analyze the direct effect of a voting rule on
the polarization caused by a chosen committee,
while the aforementioned works analyze
how preferences evolve based on a given rule.

\citet{DelJanKacSzu-2024-ConflictingPair} pursue
a goal that can be seen as opposite to ours:
selecting a most polarizing committee of size $2$;
they focus on ordinal preferences.
In a similar vein, \citet{ColGraHidMacNav-2023-Polarization} proposed measures of how \emph{divisive}, or polarizing, a single candidate is.

\section{Model}
We start by introducing key notation and proposing two ways of measuring how well a committee interconnects the voters.
For a positive integer $k \in \mathbb{N}$, define $[k]:=\{1,\dots,k\}$.

\subsection{Approval-Based Multiwinner Voting}
We consider the standard setting of approval-based multiwinner voting \citep{LaSk22b}. Given a set of $m$ candidates $C$,
an {\em election instance} ${\mathcal E} = (V, A, k)$
consists of a set of $n$ voters $V$,
an approval profile $A= (A_v)_{v\in V}$ with $A_v\subseteq C$ for all $v\in V$,
and a target committee size $k\in [m]$.
For a voter $v\in V$, the set $A_v$ captures the candidates approved by~$v$.
Throughout the paper, we view a profile $A$ as a hypergraph with vertex set $V$,
and, for each $c \in C$, a hyperedge $V_c = \{v\in V \colon c\in A_v\}$.
In the remainder of this section, we consider an election instance ${\mathcal E} = (V, A, k)$ over a candidate set~$C$.

Besides the general setting, we also consider elections with spatial one-dimensional preferences, i.e., elections where all voters and all candidates can be mapped to intervals on the real line so that a voter approves a candidate if and only if their respective intervals intersect.
Formally, following \citet{GodBatSkoFal-2021-2DApprovals}, we say that an election $(V, A, k)$ belongs to the
\emph{voter-candidate interval (VCI) domain}
if there exist a collection of positions $\{x_c\}_{c\in C}\cup \{x_v\}_{v\in V}\subseteq {\mathbb R}$ and a collection of nonnegative radii
$\{r_c\}_{c\in C}\cup \{r_v\}_{v\in V}\subseteq {\mathbb R}^+\cup\{0\}$
such that for all $v\in V$, $c\in C$ it holds that $c\in A_v$ if and only if $\lvert x_c - x_v \rvert \le r_c+r_v$.

The VCI domain is the most general domain of one-dimensional approval preferences considered in the literature. In particular, it generalizes the voter interval (VI) and candidate interval (CI) domains, defined as follows~\citep{EL15}.
An election belongs to the \emph{voter interval (VI) domain}
if there is an ordering of the voters $v_1,\dots, v_n$ such that each candidate is approved by some interval of this ordering, i.e., for each $c\in C$ there exist $i, j\in [n]$ such that $V_c =\{v_i,\dots, v_j\}$.
Similarly, an election belongs to the \emph{candidate interval (CI) domain} if there is an ordering of the candidates $c_1,\dots,c_m$ such that
for each $v\in V$ there exist $i, j\in [m]$ such that $A_v = \{c_i,\dots, c_j\}$.
It is easy to see that the VI and CI domains are contained in the VCI domain.\footnote{For instance, given an election ${\mathcal E}=(V, A, k)$ in VI, as witnessed by voter ordering $v_1, \dots, v_n$, we can set $x_{v_i}=i$ and $r_{v_i}=0$ for each $i\in [n]$. To position the candidates, for each $c\in C$ we compute $c^-=\min\{i: c\in A_{v_i}\}$ and $c^+=\max\{i: c\in A_{v_i}\}$ and set $x_c=(c^-+c^+)/2$, $r_c=(c^+-c^-)/2$. Clearly, these positions and radii certify that $\mathcal E$ belongs to the VCI domain. For CI, the construction is similar.}

A \emph{committee} for an instance $(V, A, k)$ is a subset $W \subseteq C$ with $\lvert W \rvert \le k$; we say that $W$ is {\em feasible} if $|W|=k$.
A \emph{(multiwinner) voting rule} $f$ takes as input an instance $(V,A,k)$
and outputs a feasible committee $f(V,A,k)$.

\subsection{Classic Committee Selection}

A popular classification of multiwinner voting rules
is in terms of
the main objective in electing the committee, with three most commonly studied objectives being {\em excellence}, {\em diversity}, and {\em proportionality} \citep{FSST17a}.

Both excellence and diversity are defined quantitatively:
each of these objectives is associated with a function
that assigns a numerical score to each committee,
with higher scores associated with better performance.
Formally, given an instance ${\mathcal E} = (V,A,k)$ and a committee $W$ with $|W|\le k$,
we define
\begin{align*}
\AV{}(W,{\mathcal E}) &:= \sum_{v \in V} |A_v \cap W|, \\
\CC{}(W,{\mathcal E}) &:= \lvert \{v \in V : A_v \cap W \neq \emptyset\} \rvert.
\end{align*}

For both objectives (as well as for the two novel objectives defined in Section~\ref{sec:interlace}) we omit $\mathcal E$ from the notation when it is clear from the context.
The quantities \AV{} and \CC{} are referred to as, respectively, the {\em approval score} and the {\em Chamberlin--Courant score} of committee $W$ in election $\mathcal E$. Intuitively,
\AV{} counts the number of approvals received by the members of $W$ and is viewed as a measure of excellence, while \CC{} counts the number of voters represented by $W$, i.e., voters who approve at least one member of $W$, and is viewed as a measure of diversity. The voting rule that outputs a feasible committee maximizing \AV{} (respectively, \CC{}) is known as the \emph{approval voting rule}
(respectively, the \emph{Chamberlin--Courant rule}\footnote{Originally, \citet{ChaCou-1983-ChamberlinCourant} proposed their rule for ordinal preferences. However, the approval variant of this rule is commonly studied in the computational social choice literature.}).

Consider a function $S$ that assigns scores to committees in a given election (e.g., $S = \AV{}$ or $S = \CC{}$).
Given $\alpha\in [0, 1]$, we say that a committee $W^*$ {\em satisfies
{\em $\alpha$-$S$}} for an election ${\mathcal E} = (V,A,k)$
if
\[
    S(W^*,{\mathcal E}) \ge
    \alpha \cdot \max_{\substack{W \subseteq C,\\ |W|=k}}
        S(W,{\mathcal E})\text.
\]

Moreover, we say that a voting rule $f$ {\em satisfies $\alpha$-$S$} if for every election $\mathcal E$ it holds that $f(\mathcal E)$ satisfies {\em $\alpha$-$S$} for $\mathcal E$.
For instance, the Chamberlin--Courant rule satisfies $1$-\CC.

A function $f:2^X\to\mathbb R$ is said to be {\em submodular} if for every pair of sets $S, T$ with $S\subset T\subset X$ and every $x\in X\setminus T$ it holds that $f(S\cup \{x\})-f(S)\ge f(T\cup\{x\})-f(T)$. It is immediate that, for a fixed election $\mathcal{E}$ over a candidate set $C$, the functions $\AV(W, \mathcal E)$ and $\CC{}(W, {\mathcal E})$ are submodular functions from $2^C$ to $\mathbb N$.
Our proofs will use the following basic fact about submodular functions
(see, e.g., the seminal work of \citet{nemhauser1978analysis}; for completeness, we provide a simple proof in the appendix).

\begin{restatable}{proposition}{submod}\label{prop:submodular}
    Let $f:2^X\to\mathbb R$ be a submodular function. For every pair of positive integers $\ell<k\le |X|$ and a set $S$ of size $k$ there exists a subset $S'\subseteq S$ of size $\ell$ with $\frac{f(S')}{\ell}\ge \frac{f(S)}{k}$.
\end{restatable}

In contrast to excellence and diversity, proportionality is typically captured by representation axioms. A prominent axiom of this type
is \emph{extended justified representation} (\EJR{})
\citep{AziBriConElkFreEtal-2017-EJR}; intuitively, it states that
sufficiently large groups of voters with similar preferences
should be appropriately represented in the selected committee.
We consider an approximate version in which the size of groups challenging their representation is scaled down by the approximation factor \citep[see, e.g.,][]{DHLS22a}.

    Given an election $(V, A, k)$ over $C$ and $\alpha\in (0, 1]$, a
    committee $W \subseteq C$
    is said to {\em satisfy $\alpha$-\EJR{}}
    if for every $\ell \in [k]$ and every subset $S\subseteq V$ such that
    $\alpha\cdot |S| \ge \frac{\ell}{k}\cdot|V|$ and
    $\left|\bigcap_{i \in S} A_i\right| \ge \ell$
    there exists at least one voter $i \in S$
    such that $|W \cap A_i | \ge \ell$.
    We say that a rule $f$ {\em satisfies $\alpha$-\EJR{}}
    if for every election $\mathcal E$ it holds that $f({\mathcal E})$ satisfies $\alpha$-\EJR{}.
    By setting $\alpha$ to~$1$, we obtain the standard EJR axiom.

\subsection{Interlacing Committee Selection}\label{sec:interlace}

We now define two new objectives,
which assess committees based on how well they interlace voters.

Our first objective is the number of \emph{pairs} of voters
that jointly approve a selected candidate.
Given an election ${\mathcal E}=(V, A, k)$, let $\allpairs := \{\{u,v\} \subseteq V \colon u \neq v\}$ be the set of all voter pairs. Then for each $W$ with $|W|\le k$ we set
\[
    \Pairs{}(W,{\mathcal E}) := \lvert
        \lbrace \{u,v\} \in \allpairs \colon
        A_u \cap A_v \cap W \neq \emptyset \rbrace
    \rvert.
\]

Note that for every instance ${\mathcal E} = (V,A,k)$
one can define the \emph{associated pair instance}
${\mathcal E}^{(2)} = (\allpairs, A^{(2)}, k)$, where
$A^{(2)}_{\{u,v\}} = A_u\cap A_v$
for every $\{u,v\}\in \allpairs$.
For each instance $\mathcal E$ and committee $W \subseteq C$ we have
$\Pairs{}(W,{\mathcal E}) = \CC(W,{\mathcal E}^{(2)})$.
Moreover, $\Pairs(W, {\mathcal E})$ is a submodular function from $2^C$ to $\mathbb N$.

While the \Pairs{} objective only considers
direct links between voters,
our second objective takes into account indirect connections as well.
Given an instance ${\mathcal E}=(A, V, k)$ and a subset of candidates $W \subseteq C$, we say that two voters $u, v \in V$
are \emph{connected by $W$} (and write $u \sim_W v$)
if there is a sequence of voters
$u = v_0, v_1, \dots, v_{s} = v$
with $A_{v_{i-1}} \cap A_{v_{i}} \cap W\neq \emptyset$
for every $i \in [s]$.
To evaluate a committee $W$, we
count pairs of voters
connected by $W$.
Formally,
\[
    \Connections{}(W,{\mathcal E}) :=
    \left\lvert \{\{u,v\}\in \allpairs \colon u\sim_W v\} \right\rvert.
\]

Since both \Pairs{} and \Connections{} assign scores to committees, we also consider their approximate versions, i.e., $\alpha$-\Pairs{} and $\alpha$-\Connections{}.

Our interest in \Connections{} is motivated by the following example.

\begin{example}
\label{example:connections}
    Consider a profile with six voters $v_1,\dots,v_6$,
    six \emph{cycle} candidates $c_1,\dots,c_6$,
    and two \emph{diagonal} candidates $d_1$ and $d_2$, whose hypergraph is depicted in \Cref{fig:example:connections}.
    Each cycle candidate is approved by two consecutive voters:
    for $i=1, \dots, 5$ candidate
    $c_i$ is approved by $v_i$ and $v_{i+1}$, while $c_6$ is approved by $v_1$ and $v_6$.
    Also, $d_1$ is approved by $v_2$ and $v_6$
    and $d_2$ by $v_3$ and $v_5$. Let $k=6$.

    Consider the following two committees: $W = \{c_1,c_2,c_3,c_4,c_5,c_6\}$ contains all cycle candidates,
    whereas in $W' = \{c_1,c_3,c_4,c_6,d_1,d_6\}$
    two cycle candidates are exchanged for the diagonal candidates
    ($W'$ is marked by thick red hyperedges
    in \Cref{fig:example:connections}).
    Common voting rules, including
    the approval voting rule and the Chamberlin--Courant rule,
    do not distinguish between $W$ and $W'$, as
    each voter approves exactly two candidates in either committee.
    Moreover, the rule that maximizes \Pairs{}
    is also unable to distinguish them,
    as both $W$ and $W'$ cover exactly 6 pairs of voters.
    However, intuitively, $W'$ seems more polarizing: under $W'$, there are two disconnected groups of voters,
    each supporting (though not fully) their own set of candidates.

    In contrast, a rule
    that maximizes \Connections{} is sensitive to the differences between the two committees.
    Under $W$, all $15$ pairs of voters are connected,
    while $W'$ only achieves $6$ connections.
    \hfill$\lhd$
\end{example}

Note that, in contrast to \AV{}, \CC{} and \Pairs{}, the \Connections{} objective is not submodular: e.g., in Example~\ref{example:connections} adding $c_2$ to $\{c_1\}$ creates two additional connected pairs, while adding $c_2$ to $\{c_1, c_3\}$ creates four additional connected pairs.

    \begin{figure}
        \centering

        \begin{tikzpicture}
        \def\xs{1cm}
        \def\ys{1cm}

            \node[voter] (v1) at (-2*\xs,0) {$v_1$};
            \node[voter] (v2) at (-0.8*\xs,\ys) {$v_2$};
            \node[voter] (v3) at (0.8*\xs,\ys) {$v_3$};
            \node[voter] (v4) at (2*\xs,0) {$v_4$};
            \node[voter] (v5) at (0.8*\xs,-1*\ys) {$v_5$};
            \node[voter] (v6) at (-0.8*\xs,-1*\ys) {$v_6$};
            \draw[myred, rotate = 40, scale = 1, line width = 1] ($(v1)!.5!(v2)$) ellipse (1.2cm and 0.45cm);
            \draw[scale = 1, line width = 0.5] ($(v2)!.5!(v3)$) ellipse (1.2cm and 0.4cm);
            \draw[myred, rotate = -40, scale = 1, line width = 1] ($(v3)!.5!(v4)$) ellipse (1.2cm and 0.45cm);
            \draw[myred, rotate = 40, scale = 1, line width = 1] ($(v4)!.5!(v5)$) ellipse (1.2cm and 0.45cm);
            \draw[scale = 1, line width = 0.5] ($(v5)!.5!(v6)$) ellipse (1.2cm and 0.4cm);
            \draw[myred, rotate = -40, scale = 1, line width = 1] ($(v6)!.5!(v1)$) ellipse (1.2cm and 0.45cm);
            \draw[myred, scale = 1, line width = 1] ($(v2)!.5!(v6)$) ellipse (0.4cm and 1.4cm);
            \draw[myred, scale = 1, line width = 1] ($(v3)!.5!(v5)$) ellipse (0.4cm and 1.4cm);
            \node (c1) at ($(v1)!.5!(v2)+(130:.65)$) {$c_1$};
            \node (c2) at ($(v2)!.5!(v3)+(90:.6)$) {$c_2$};
            \node (c3) at ($(v3)!.5!(v4)+(50:.65)$) {$c_3$};
            \node (c4) at ($(v4)!.5!(v5)+(310:.65)$) {$c_4$};
            \node (c5) at ($(v5)!.5!(v6)+(270:.6)$) {$c_5$};
            \node (c6) at ($(v6)!.5!(v1)+(230:.65)$) {$c_6$};
            \node (d1) at ($(v2)!.5!(v6)$) {$d_1$};
            \node (d2) at ($(v3)!.5!(v5)$) {$d_2$};

        \end{tikzpicture}
        \caption{Illustration of \Cref{example:connections}.
        When the target committee size is $6$, every size-$6$ subset of candidates maximizes the \Pairs{} objective.
        However, \Connections{} is higher for the committee $\{c_1,\dots, c_6\}$ than for the disconnected committee $\{c_1,c_3,c_4,d_1,d_2\}$ (indicated by thick red lines).}
        \label{fig:example:connections}
    \end{figure}
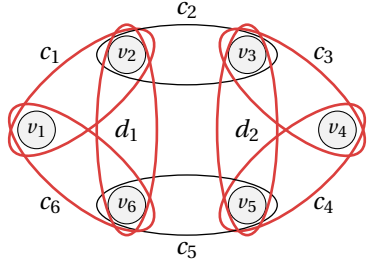

\section{Computation of The New Objectives}\label{sec:comp}
In this section, we show that maximizing \Pairs{} and \Connections{} is NP-hard in general, but
tractable on well-structured domains.

\subsection{General Preferences}\label{sec:hard}
Our hardness proofs are based on the \NP-complete problem \textsc{Exact Cover by $3$-Sets (X3C)} \citep{GaJo79a}.
An instance of \textsc{X3C} is a pair $(R,\mathcal S)$, where $R$ is a ground set of size $3\rho$ and $\mathcal S$ is a collection of $3$-element subsets of $R$; it is a Yes-instance if and only if there exists a subset $\mathcal S'\subseteq \mathcal S$ with $|{\mathcal S}'|=\rho$ that covers~$R$.

The proof idea for the \Pairs{} objective is to represent every element in the ground set of an \textsc{X3C} instance by a pair of voters.

\begin{restatable}{theorem}{PairsHardness}
\label{thm:pairs:hardness}
    It is \NP-complete to decide, given an election ${\mathcal E}=(V,A,k)$ and a threshold $q\in\mathbb N$, whether there exists a committee $W$ of size at most $k$ such that $\Pairs{}(W,{\mathcal E})\ge q$.
\end{restatable}

\begin{proof}
    Membership in {\NP} is immediate: for a given committee, its size and the number of pairs of voters approving a common candidate can be checked in polynomial time.

    To show {\NP}-hardness, we present a reduction from {\sc X3C}.
    Given an instance $(R,\mathcal S)$ of \textsc{X3C} with $|R|=3\rho$, we construct an instance of our problem as follows. We create one candidate for each set in $\mathcal S$ and two voters for each element of the ground set, i.e., we set
    $C = \{c_S\colon S\in \mathcal S\}$ and
    $V = \{v_r,v'_r\colon r\in R\}$.
    For each $S\in \mathcal S$, candidate $c_S$ is approved by voters $\{v_r,v'_r\colon r\in S\}$.
    We set the target committee size $k$ to $\rho$ and the threshold $q$ to $15\rho$.
    We will show that we can cover $q$ pairs of voters if and only if the source instance is a Yes-instance of \textsc{X3C}.

    Suppose first there exists a feasible committee $W$ that covers $q$ pairs of voters. Each $c\in W$ is approved by exactly $6$ voters, so it can cover at most $\binom{6}{2} = 15$ pairs of voters. Moreover, the candidates' support sets are either disjoint or overlap in at least two voters. As $q=15k$, this means that candidates in $W$
    have pairwise disjoint support sets. Since $|W|=k$, it follows that
    $\{S\in {\mathcal S}: c_S\in W\}$ forms a cover of $R$, i.e., our instance of {\sc X3C} is a Yes-instance.

    Conversely, assume that there exists a subset $\mathcal S'\subseteq \mathcal S$ of size~$k$ that covers $R$.
    Consider the committee $W = \{c_S\colon S\in \mathcal S'\}$.
    Then, $|W| = |\mathcal S'| = \rho = k$.
    Moreover, since, all of the sets in $\mathcal S'$ are pairwise disjoint, the support sets of the candidates in $W$ are pairwise disjoint and
    contain exactly $6$ voters each.
    Hence, there are $k\cdot \binom{6}{2} = 15\rho = q$ pairs of voters who approve a common candidate.
\end{proof}

A similar hardness result holds for \Connections.
The proof idea is to introduce an auxiliary voter that is the focal point in connecting all voters.

\begin{restatable}{theorem}{HardnessConn}
\label{thm:connected:hardness}
    It is \NP-complete to decide, given an election ${\mathcal E}=(V,A,k)$ and a threshold $q\in\mathbb N$, whether there exists a committee $W$ of size $k$ such that
    $\Connections(W, {\mathcal E})\ge q$. The hardness result holds even if $q=\binom{n}{2}$,
    i.e., if the goal is to connect all $n$ voters.
\end{restatable}

\begin{proof}
    Membership in {\NP} is immediate: given an election instance $(V, A, k)$, a committee $W$ and a target $q$, we (1) construct a graph $G$ with vertex set $V$
    where there is an edge between $v$ and $v'$ if and only if there is a candidate $c\in W$
    that is approved by both $v$ and $v'$; (2)
    identify the connected components of $G$, which we denote by $G_1, \dots, G_t$; (3) return ``yes'' if and only if $\sum_{s\in [t]}\frac{|G_s|(|G_s|-1)}{2}\ge q$, where $|G_s|$ denotes the number of vertices of $G_s$.

    To show {\NP}-hardness, we present a reduction from {\sc X3C}.
    Given an instance $(R,\mathcal S)$ of \textsc{X3C} with $|R|=3\rho$, we construct an instance of our problem as follows. We create a candidate for each set in $\mathcal S$ and a voter for each element of the ground set $R$, as well as one additional voter, i.e., we set
    $C = \{c_S\colon S\in\mathcal S\}$,
    $V = \{v\}\cup \{v_r\colon r\in R\}$.
    For each $S\in \mathcal S$, $c_S$ is approved by $\{v\}\cup \{v_r\colon r\in S\}$.
    We want to select a committee $W\subseteq C$ of size $k = \rho$ and set the threshold $q$ to $\binom{n}{2}$, where $n = |V|$.

    Consider a collection ${\mathcal S}'\subseteq \mathcal S$ of size $k$ and the respective committee $W=\{c_S: S\in {\mathcal S'}\}$. If ${\mathcal S}$ covers $R$, each voter in $V$ approves a candidate in $W$, and $v$ approves all candidates,
    so all voters are connected via $v$. Conversely, if all pairs of voters are connected, then each voter must approve some candidate in $W$ and hence ${\mathcal S}'$ covers $R$.
    This completes the proof.
\end{proof}

\subsection{One-Dimensional Preferences}
We will now complement our hardness results by arguing that,
on elections in the VCI domain, we can maximize \Pairs{} and \Connections{} in polynomial time.

We start by observing that, for the objectives we consider, a VCI instance can be transformed into a CI instance without changing the value of these objectives. To this end, we define a notion of dominance among candidates and prove that, in the absence of dominated candidates, every VCI instance is a CI instance.

\subsubsection{From VCI to CI}
Given an election ${\mathcal E}=(V, A, k)$ over a candidate set $C$, we say that a candidate $c'\in C$ is {\em dominated} by a candidate $c\in C$ if every voter approving~$c'$ also approves~$c$,
and some voter approves $c$ but not $c'$, i.e.,
$V_{c'}$ is a proper subset of~$V_c$.

It turns out that if an election in the VCI domain contains no dominated candidates, it belongs to the much simpler to analyze CI domain; this observation, which
is implicit in the work of \citet[][Lemma 4.7]{ElkFalIgaManSchEtal-2024-PriceOfJR},
may be of independent interest.
Indeed, removal of dominated candidates from a winning committee does not affect the \Pairs{} and \Connections{} objectives, so we can simply remove all dominated candidates from the input instance.

\begin{restatable}{proposition}{VCItoCI}\label{prop:VCI-to-CI}
    Let $\mathcal E$ be an instance in the VCI domain.
    If~$\mathcal E$ contains no dominated candidates, then it belongs to the CI domain.
\end{restatable}

\begin{proof}
    Consider an election ${\mathcal E}=(V, A, k)$ over the candidate set $C$ that belongs to the VCI domain, as witnessed by positions
    $\{x_c\}_{c\in C}\cup \{x_v\}_{v\in V}$ and radii
    $\{r_c\}_{c\in C}\cup \{r_v\}_{v\in V}$.
    Renumber the candidates so that
    $x_{c_1} \le x_{c_2}\le\dots\le x_{c_m}$.

    Suppose for the sake of contradiction that
    this ordering of the candidates does not witness that $\mathcal E$ belongs to CI.
    Then, there exists a voter $v\in V$ and $h<i<j$ such that $v$ approves $c_h$ and $c_j$, but not $c_i$.
    For readability, we will refer to the positions and radii of $c_h$, $c_i$ and $c_j$ as $x_h, x_i, x_j$ and $r_h, r_i, r_j$, respectively. Since $v$ does not approve $c_i$, we have $x_v\neq x_i$; we can then assume without loss of generality that $x_v< x_i\le x_j$.
    To obtain a contradiction, we will show that $c_i$ is dominated by $c_j$.

    To this end, we will argue that $[x_i-r_i, x_i+r_i]\subseteq [x_j-r_j, x_j+r_j]$.
    Indeed, $c_i\not\in A_v$ implies
    $x_i- r_i > x_v +r_v$, whereas
    $c_j\in A_v$ implies
    $x_j- r_j \le x_v +r_v$.
    Combining these inequalities,
    we obtain $x_i-r_i>x_j-r_j$.
    It then follows that
    $x_i+r_i<x_i+(x_i-x_j+r_j)\le x_j+r_j$,
    where the last inequality follows from $x_i\le x_j$.
    Thus, the interval of $c_i$ is subsumed by that of $c_j$, and hence every voter who approves $c_i$ also approves $c_j$. Moreover, $v$ approves $c_j$, but not $c_i$. We have shown that $c_i$ is dominated, concluding the proof.
\end{proof}

In what follows, we state our results for the VCI domain, but assume that the input election belongs to the CI domain, and we are explicitly given the respective candidate order. It will also be convenient to assume that this order is $c_1, \dots, c_m$. This requires two preprocessing steps: first, we eliminate all dominated candidates (which, by \Cref{prop:VCI-to-CI}, results in a CI election), and second, we compute an ordering of the candidates witnessing that our instance belongs to the CI domain.
Both steps can be implemented in polynomial time \citep[for the second step, see, e.g.,][]{EL15}.

\subsubsection{Efficient Algorithms}We are ready to present polynomial-time algorithms for \Pairs{} and \Connections{} on the VCI domain.
Since \Pairs{} is identical to \CC{} on the associated pair instance, we can compute \Pairs{} by leveraging an existing algorithm for \CC{} in the CI domain~\citep{BSU13,EL15}.

\begin{restatable}{proposition}{VCIpairs}\label{prop:pairs-VCI}
    In the VCI domain, a committee that maximizes \Pairs{} can be computed in polynomial time.
\end{restatable}

\begin{proof}
    Fix an election instance $\mathcal E$.
    As argued earlier, we can assume that $\mathcal E$ is in the CI domain with respect to candidate ordering $c_1,\dots, c_m$.
    Recall that
    $\Pairs{}(W,{\mathcal E}) = \CC(W,{\mathcal E}^{(2)})$.
    Now, note that for all $\{u,v\}\in \allpairs$, it holds by definition that $A_{\{u,v\}} = A_u\cap A_v$ is the intersection of two intervals of that ordering and hence itself an interval. Thus, ${\mathcal E}^{(2)}$ is in the CI domain with respect to the same candidate ordering.
    For instances in the CI domain, \CC{} can be maximized in polynomial time~\citep{BSU13,EL15}.
\end{proof}

In the VCI domain, we can also compute a committee that maximizes \Connections{} in polynomial time;
however, the argument is significantly more complicated.
Again, we assume that the input profile belongs to the CI domain, as witnessed by the candidate ordering
$c_1,\dots,c_m$.
A natural idea, then, is to use dynamic programming to compute, for each $b\in [k]$ and $i\in [m]$, an optimal subcommittee of size $b$ with rightmost candidate $c_i$.
For $b=1$, the computation is straightforward, and for $b=k$, one of the resulting $m$ committees globally maximizes \Connections.
However, computing the value of adding~$c_i$ to a committee of size $b-1$ that has $c_j$, $j<i$, as its rightmost candidate is a challenging task: this is because the number of connections that $c_i$ adds depends on the size of the connected component associated with $c_j$. To handle this, we add a third dimension to the dynamic program: the number of voters $x\in [n]$ in the connected component of the last selected candidate.
The resulting dynamic program has $\orderof(mnk)$ cells, and each cell can be filled in polynomial time given the values of the already-filled-in cells.
We present the proof in \Cref{app:DP}.

\begin{restatable}{theorem}{VCIconn}\label{thm:vci-conn}
In the VCI domain, a committee that maximizes \Connections{} can be computed in polynomial time.
\end{restatable}

\section{Combining Objectives}
\label{sec:trade-offs}

While interlacing objectives can be viewed in isolation, in many cases, standard objectives of excellence, diversity, or proportionality continue to be important for the selection of a committee.
In this section, we investigate to what extent we can select committees that simultaneously perform well with respect to both interlacing and standard objectives.
Many of our results will show that there are inherent trade-offs between the objectives.
For this purpose, given two objectives, we construct instances whose hypergraph representation can be partitioned into two components, each corresponding to one objective.By design, candidates selected in one component mainly contribute to this component's objective while having a negligible effect on the other objective.

\subsection[Pairs Objective]{\Pairs{} Objective}
\label{sec:trade-offs:pairs}

First, we consider combining the \Pairs{} objective
with individual excellence of the committee members, as measured by \AV{}.
For every $\alpha \in [0,1]$ and
every election ${\mathcal E} = (V, A, k)$,
there is a simple way to obtain a simultaneous
$\lceil \alpha k \rceil / k $-approximation of \Pairs{}
and $\lfloor (1- \alpha) k \rfloor / k $-approximation of \AV{}.
Indeed, we can split the $k$ positions on the committee into two parts
of size $k_1 =\lceil \alpha k \rceil$ and
$k_2 = \lfloor (1- \alpha) k \rfloor$, respectively, and then
select $k_1$ candidates so as to maximize \Pairs{}
and $k_2$ candidates so as to maximize \AV{}
(if some candidate is selected both times, we replace their second copy by an arbitrary
unselected candidate).
Since \Pairs{} and \AV{} are submodular functions, \Cref{prop:submodular} implies that
this procedure provides the desired guarantees.
We can use the same technique to combine \Pairs{} with the goal of diverse representation, as measured by \CC{} (recall that \CC{} is submodular, too).
Note that \citet{LacSko-2020-MultiwinnerApproximations} propose a similar method for combining \AV{} and \CC{}.

\begin{proposition}
    \label{prop:av_cc_pair_app:lb}
    For every $\alpha \in [0,1]$ and election $\mathcal E$,
    there exist committees that satisfy
    \begin{enumerate}
        \item $\lceil \alpha k \rceil / k $-\Pairs{}
    and $\lfloor (1- \alpha) k \rfloor / k $-\AV,
    \item $\lceil \alpha k \rceil / k $-\Pairs{}
    and $\lfloor (1- \alpha) k \rfloor / k $-\CC.
    \end{enumerate}
\end{proposition}

Perhaps surprisingly, it turns out that, for both combinations, this is the best we can hope for.

\begin{proposition}
    \label{prop:av_pair_app:ub}
    For every $\alpha, \beta \in [0, 1]$,
    if a voting rule satisfies
    $\alpha$-\Pairs{}
    and $\beta$-\AV{},
    then $\alpha + \beta \le 1$.
\end{proposition}
    \begin{figure}
        \centering

        \begin{tikzpicture}

            \tikzset{
            voter/.style={circle,draw,minimum size=0.1cm,inner sep=0, fill = black!05, font=\footnotesize},
          }

            \draw [decorate,thick,decoration={brace,amplitude=5pt,mirror,raise=4ex}]
          (0.35,0.3) -- (0.35,-2.3) node[midway,xshift=-3.4em,yshift=1.5pt]{$x^3$};

            \draw [decorate,thick,decoration={brace,amplitude=5pt,mirror,raise=4ex}]
      (1.1,-0.2) -- (-0.1,-0.2) node[midway,yshift=3em]{$x$};

            \node[voter] (v1) at (0,0) {};
            \node[voter] (v2) at (0.2,0) {};
            \node[voter] (v3) at (0.4,0) {};
            \node at (0.7, 0) {$...$};
            \node[voter] (v4) at (1,0) {};
            \draw[scale = 1] ($(v1)!.5!(v4)$) ellipse (0.6cm and 0.3cm);
            \node at (1.4, 0) {$b_1$};

            \node[voter] (v1) at (0,-0.8) {};
            \node[voter] (v2) at (0.2,-0.8) {};
            \node[voter] (v3) at (0.4,-0.8) {};
            \node at (0.7, -0.8) {$...$};
            \node[voter] (v4) at (1,-0.8) {};
            \draw[scale = 1] ($(v1)!.5!(v4)$) ellipse (0.6cm and 0.3cm);
            \node at (1.4, -0.8) {$b_2$};

            \node at (0.5, -1.4) {$\vdots$};

            \node[voter] (v1) at (0,-2) {};
            \node[voter] (v2) at (0.2,-2) {};
            \node[voter] (v3) at (0.4,-2) {};
            \node at (0.7, -2) {$...$};
            \node[voter] (v4) at (1,-2) {};
            \draw[scale = 1] ($(v1)!.5!(v4)$) ellipse (0.6cm and 0.3cm);
            \node at (1.45, -2) {$b_{x^3}$};

            \draw [decorate,thick,decoration={brace,amplitude=5pt,mirror,raise=4ex}]
          (3.45,-0.3) -- (3.45,-1.7) node[midway,xshift=-3em,yshift=1.5pt]{$x^2$};

            \node[voter] (v1) at (3.5,-0.4) {};
            \node[voter] (v2) at (3.5,-0.6) {};
            \node[voter] (v3) at (3.5,-0.8) {};
            \node[voter] (v4) at (3.5,-1.6) {};
            \node at ($(v3)!.5!(v4)$) {$\vdots$};

            \node[ellipse,draw,minimum width=0.5cm,minimum height=1.5cm] (c1) at (3.5,-1) {};

            \node[ellipse,draw,minimum width=1.2cm,minimum height=1.75cm] (c2) at (3.65,-1) {};

            \node[ellipse,draw,minimum width=2.3cm,minimum height=2.2cm] (c3) at (4,-1) {};
            \node at (3.97, -1) {$c_1$};
            \node at (4.67, -1) {$c_2 \  \cdots$};
            \node at (5.65, -1.01) {$c_{x^3+1}$};
        \end{tikzpicture}

        \caption{Illustration of the profile constructed in the proof of \Cref{prop:av_pair_app:ub}.
        The block voters are on the left and the central voters are on the right.
        Each block candidate is approved by $x$ block voters, whereas each central candidate is approved by all central voters.}
        \label{fig:prop:av_pair_app:ub}
    \end{figure}
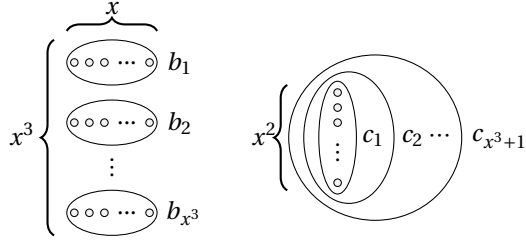

\begin{proof}
    We will construct a family of instances that allows us to bound the sum of approximation ratios.
    For a given constant $x \in \mathbb{N}$,
    consider the election $\mathcal{E} = (V, A, k)$
    defined as follows
    (see \Cref{fig:prop:av_pair_app:ub} for an illustration).
    The set $C$ consists of $x^3$ \emph{block candidates}
    $b_1, \dots, b_{x^3}$,
    and $x^3+1$ \emph{central candidates} $c_1, \dots, c_{x^3+1}$, so that $|C|=2x^3+1$.
    The set $V$ consist of $x^3$ groups of {\em block voters} $(B_i)_{i\in [x^3]}$ of size $x$ each,
    and a single group of $x^2$ \emph{central voters}.
    For $i\in [x^3]$, each voter in block $B_i$ approves candidate $b_i$ only, whereas each central voter approves all central candidates $c_1, \dots, c_{x^3+1}$.
    The target committee size is set to $k=x^3 +1$.

    Since there are fewer than $k$ block candidates, every committee $W$ contains at least one central candidate, who is approved by all $x^2$ central voters,
    and covers all $(x^2-1)x^2/2$ pairs of central voters.
    Then, every additional central candidate contributes $x^2$ to the AV objective and $0$ to the \Pairs{} objective, whereas every additional block candidate contributes
    $x$ to the AV objective and $x(x-1)/2$ to the \Pairs{} objective.

    Assume that for some
    $\gamma \in \{0, 1/x^3, 2/x^3, \dots, 1\}$
    our rule selects a committee $W_\gamma \subseteq C$ with $\gamma x^3 + 1$ central candidates
    and $(1-\gamma)x^3$ block candidates.
    Then,
    we obtain the following \AV{} and \Pairs{} scores:
    \begin{align*}
        \AV{}(W_\gamma, \mathcal{E}) &= (\gamma x^3+1)x^2+(1-\gamma) x^4 = \gamma x^5 + \mathcal{O}(x^4),
            \mbox{ and}\\
        \Pairs{}(W_\gamma, \mathcal{E}) &= \frac{x^4 - x^2}{2} +(1-\gamma)\frac{x^5 - x^4}{2} = \frac{1-\gamma}{2} x^5 + \mathcal{O}(x^4),
    \end{align*}
    where the ${\mathcal O}(\cdot)$ upper bound holds for all values of $\gamma$.
    Observe that the maximum \AV{} score is obtained when we take $\gamma=1$,
    and the maximum \Pairs{} score is obtained when $\gamma = 0$.
    Also,
    \[
        \frac{\AV{}(W_\gamma, \mathcal{E})}{\AV{}(W_1, \mathcal{E})} + \frac{\Pairs{}(W_\gamma, \mathcal{E})}{\Pairs{}(W_0, \mathcal{E})} \le 1 + \orderof(1/x).
    \]

    Hence, when $x$ tends to infinity,
    the sum of approximation ratios for \AV{} and \Pairs{} gets arbitrarily close
    to~$1$.
\end{proof}

One may expect the \Pairs{} and \CC{} objectives to be more aligned than \Pairs{} and \AV{}.
Indeed, recall that \Pairs{} is the same as computing \CC{} on the associated pair instance.
However, surprisingly,
the worst-case trade-off for these objectives is the same as for \Pairs{} and \AV.

\begin{restatable}{proposition}{PairsVsCC}
    \label{prop:cc_pair_app:ub}
    For every $\alpha, \beta \in [0, 1]$,
    if a voting rule satisfies $\alpha$-\Pairs{}
    and $\beta$-\CC,
    then $\alpha + \beta \le 1$.
\end{restatable}
    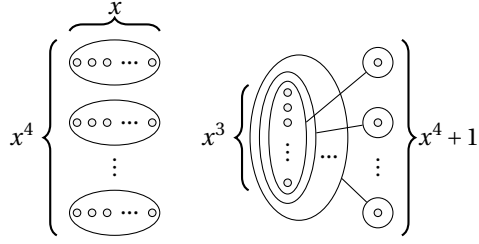
\begin{figure}
        \centering

        \begin{tikzpicture}

            \tikzset{
            voter/.style={circle,draw,minimum size=0.1cm,inner sep=0, fill = black!05, font=\footnotesize},
          }

            \draw [decorate,thick,decoration={brace,amplitude=5pt,mirror,raise=4ex}]
          (0.35,0.3) -- (0.35,-2.3) node[midway,xshift=-3.5em,yshift=1.5pt]{$x^4$};

            \draw [decorate,thick,decoration={brace,amplitude=5pt,mirror,raise=4ex}]
      (1.1,-0.2) -- (-0.1,-0.2) node[midway,yshift=3em]{$x$};

            \node[voter] (v1) at (0,0) {};
            \node[voter] (v2) at (0.2,0) {};
            \node[voter] (v3) at (0.4,0) {};
            \node at (0.7, 0) {$...$};
            \node[voter] (v4) at (1,0) {};
            \draw[scale = 1] ($(v1)!.5!(v4)$) ellipse (0.6cm and 0.3cm);

            \node[voter] (v1) at (0,-0.8) {};
            \node[voter] (v2) at (0.2,-0.8) {};
            \node[voter] (v3) at (0.4,-0.8) {};
            \node at (0.7, -0.8) {$...$};
            \node[voter] (v4) at (1,-0.8) {};
            \draw[scale = 1] ($(v1)!.5!(v4)$) ellipse (0.6cm and 0.3cm);

            \node at (0.5, -1.4) {$\vdots$};

            \node[voter] (v1) at (0,-2) {};
            \node[voter] (v2) at (0.2,-2) {};
            \node[voter] (v3) at (0.4,-2) {};
            \node at (0.7, -2) {$...$};
            \node[voter] (v4) at (1,-2) {};
            \draw[scale = 1] ($(v1)!.5!(v4)$) ellipse (0.6cm and 0.3cm);

            \draw [decorate,thick,decoration={brace,amplitude=5pt,mirror,raise=4ex}]
          (2.9,-0.3) -- (2.9,-1.7) node[midway,xshift=-3.5em,yshift=1.5pt]{$x^3$};

            \node[voter] (v1) at (2.8,-0.4) {};
            \node[voter] (v2) at (2.8,-0.6) {};
            \node[voter] (v3) at (2.8,-0.8) {};
            \node at (2.8, -1.2) {$\vdots$};
            \node[voter] (v4) at (2.8,-1.6) {};

            \node[ellipse,draw,minimum width=0.5cm,minimum height=1.5cm] (c1) at (2.8,-1) {};
            \node[voter] (v1) at (4, 0) {};
            \node[ellipse,draw,minimum width=0.4cm,minimum height=0.4cm] (c1_) at (4, -0) {};

            \node[ellipse,draw,minimum width=0.75cm,minimum height=1.75cm] (c2) at (2.8,-1) {};
            \node[voter] (v2) at (4, -0.8) {};
            \node[ellipse,draw,minimum width=0.4cm,minimum height=0.4cm] (c2_) at (4, -0.8) {};

            \node at (4, -1.4) {$\vdots$};

            \node[ellipse,draw,minimum width=1.3cm,minimum height=2.2cm] (c3) at (2.95,-1) {};
            \node at (3.35, -1.25) {$...$};
            \node[voter] (v2) at (4, -2) {};
            \node[ellipse,draw,minimum width=0.4cm,minimum height=0.4cm] (c3_) at (4, -2) {};

            \path[draw]
            (c1) edge (c1_)
            (c2) edge (c2_)
            (c3) edge (c3_)
            ;

            \draw [decorate,thick,decoration={brace,amplitude=5pt,mirror,raise=4ex}]
          (3.7,-2.3) -- (3.7,0.3) node[midway,xshift=4em,yshift=1.5pt]{$x^4 + 1$};
        \end{tikzpicture}

        \caption{An illustration of the profile constructed in the proof of \Cref{prop:cc_pair_app:ub}.
        Block voters are on the left, central voters in the middle, and arm voters on the right.
        Each block candidate is approved by $x$ block voters, whereas each central candidate is approved by all central voters and one arm voter.}
        \label{fig:prop:cc_pair_app:ub}
    \end{figure}

\begin{proof}
    For a given constant $x \in \mathbb{N}$,
    consider an election instance $\mathcal{E} = (V, A, k)$
    defined as follows
    (see \Cref{fig:prop:cc_pair_app:ub} for an illustration).

    The candidate set consists of $x^4$ \emph{block candidates}
    $b_1,\dots,b_{x^4}$
    and $x^4 + 1$ \emph{arm candidates}
    $a_1,\dots,a_{x^4+1}$,
    i.e., $2x^4 + 1$ candidates in total.
    The set $V$ consists of $x^5$ {\em block voters}, split into $x^4$ blocks $B_1, \dots, B_{x^4}$
    of size $x$ each,
    $x^4 + 1$ \emph{arm voters},
    and $x^3$ \emph{central voters},
    i.e., $x^5 + x^4 + x^3 + 1$ voters in total.

    For each $i \in [x^4]$ the voters in block
    $B_i$ approve the block candidate $b_i$.
    For each $i\in [x^4+1]$, the $i$th arm voter
    approves the $i$th arm candidate $a_i$,
    and all central voters approve all arm candidates.

    The target committee size is set to $k=x^4 +1$.
    As there are fewer than $k$ block candidates, every committee of size $k$ contains at least one arm candidate. Thus,
    by symmetry, without loss of generality,
    we can assume that for some
    $\gamma \in \{0, 1/x^4, 2/x^4, \dots, 1\}$
    we select a committee $W_\gamma \subseteq C$
    that consists of arm candidates $a_1, a_2, \dots, a_{\gamma x^4 + 1}$
    and block candidates $b_1, b_2, \dots, b_{(1-\gamma)x^4}$.
    Every selected block candidate covers
    $x$ voters and $x(x-1)/2$ pairs of voters.
    Moreover, the first selected arm candidate covers
    $x^3+1$ voters and $x^3(x^3+1)/2$ pairs of voters, whereas
    every subsequent arm candidate covers
    $1$ voter and $x^3$ pairs of voters.
    Thus, when we select $\gamma x^4 + 1$ arm candidates and
    $(1-\gamma)x^4$ block candidates,
    we obtain the following \CC{} and \Pairs{} scores:
    \begin{align*}
        \CC{}(W_\gamma,\mathcal{E}) &= (1-\gamma)x^5 + \gamma x^4 + x^3 + 1,
            \mbox{ and}\\
        \Pairs{}(W_\gamma,\mathcal{E}) &= \gamma x^7
            + \frac{2 - \gamma}{2}x^6
            - \frac{1 - \gamma}{2}x^5
            + \frac{1}{2} x^3.
    \end{align*}
    Observe that the maximum number of covered voters is obtained when we take $\gamma=0$, whereas the maximum number of covered pairs of voters is obtained when $\gamma = 1$.
    Also, we have
    \[
        \frac{\CC{}(W_\gamma,\mathcal{E})}{\CC{}(W_0,\mathcal{E})} +
        \frac{\Pairs{}(W_\gamma,\mathcal{E})}{\Pairs{}(W_1,\mathcal{E})} = 1 + \orderof(1/x).
    \]

    Hence, when $x$ tends to infinity,
    the sum of approximation ratios for \CC{} and \Pairs{} gets arbitrarily close
    to~$1$.
\end{proof}

Finally, we investigate
how to combine the \Pairs{} objective
with proportional representation,
as captured by the \EJR{} axiom.
Again, we can use the committee-splitting technique to
show that for every election $\mathcal{E}=(V,A,k)$
there is a committee that satisfies
$\lceil \alpha k \rceil/k$-\Pairs{} and
$(1-\alpha)$-\EJR{}.
For this, we first need to show that
we can guarantee $(1-\alpha)$-\EJR{} with a $(1-\alpha)$-fraction of the committee seats.
To this end, we employ a variant of the \emph{method of equal shares} (MES) by \citet{PeSk20a}.
Briefly, this rule gives each voter $k/n$ units of money; it then sequentially selects candidates that are best for voters that still have money, and subtracts money from the supporters of the selected candidates.
By adapting the proof that MES satisfies \EJR{} \citep{PeSk20a},
we show that, by executing MES while scaling the voters' budgets by $\alpha$ (we will refer to this rule as $\alpha$-MES),
we obtain $\alpha$-\EJR{} for the original instance; we believe that this result
is of independent interest.\footnote{
A similar observation was made by
\citet{DoPe25a},
but they require
$\lceil \alpha k \rceil$ candidates, which in our case
would allow only for a rounded-down \Pairs{} guarantee.
}
We defer a formal definition of $(\alpha$-)MES and the proof of \Cref{lemma:equal-shares} to \Cref{app:MES}.

\begin{restatable}{lemma}{MESalpha}
        \label{lemma:equal-shares}
For every $\alpha\in(0, 1)$,
given an election $\mathcal E = (V, A, k)$,
the rule $\alpha$-MES runs in polynomial time and returns a committee of size $\lfloor \alpha k\rfloor$ that satisfies $\alpha$-\EJR{}.
\end{restatable}

We remark that the committee obtained as described in \Cref{lemma:equal-shares} satisfies an even stronger notion of proportionality, namely $\alpha$-\EJR+ \citep{BrPe23a}.
Using \Cref{lemma:equal-shares}, we now easily obtain the desired guarantees.

\begin{proposition}
    \label{prop:pair_ejr_app:lb}
    For every $\alpha \in [0,1]$ and election $\mathcal E$,
    there exists a committee that satisfies
    $\alpha$-\Pairs{} and $(1-\alpha)$-\EJR.
\end{proposition}
\begin{proof}
    Consider an election $\mathcal E$.
    By \Cref{lemma:equal-shares}, we can satisfy $(1 - \alpha)$-\EJR{} using $\lfloor (1 - \alpha) k\rfloor$ candidates.
    With the remaining $k -\lfloor (1 - \alpha) k\rfloor = \lceil \alpha k\rceil$ candidates, by \Cref{prop:submodular} we can guarantee $\alpha$-\CC{} on the associated pair instance $\mathcal E^{(2)}$. This is equivalent to satisfying $\alpha$-\Pairs{} on $\mathcal E$, concluding the proof.
\end{proof}

As before, we provide a matching upper bound.
We note that our proof works
even if, instead of \EJR{}, we consider
the much weaker axiom of \emph{justified representation} (JR) \citep{AziBriConElkFreEtal-2017-EJR}.

\begin{restatable}{proposition}{EJRvsPairsTradeoff}\label{prop:EJR_vs_Pairs:ub}
    For every $\alpha, \beta \in [0, 1]$,
    if a voting rule satisfies
    $\alpha$-\Pairs{}
    and $\beta$-EJR,
    then $\alpha + \beta \le 1$.
\end{restatable}

\begin{proof}
Clearly, the statement is true for $\beta = 0$.
Let $f$ be a voting rule that satisfies $\beta$-EJR for some $\beta \in (0,1]$.
Fix an $\varepsilon>0$ so that $\varepsilon < \beta$ and $\beta-\varepsilon\in\mathbb Q$.
We modify the election $\mathcal E$ from the proof of \Cref{prop:cc_pair_app:ub} by reducing the number of block candidates (resp., block voters) from $x^4$ to $(\beta-\varepsilon) x^4$ (resp., $(\beta-\varepsilon) x^5$) and only considering values of $x$ for which $(\beta-\varepsilon)x$ is integer
(as $\beta-\varepsilon$ is rational, there are infinitely many such values).
As in the proof of \Cref{prop:cc_pair_app:ub}, we want to select $k = x^4+1$ candidates.
Denote the resulting election by ${\mathcal E}'$.
    \begin{figure}
        \centering

        \begin{tikzpicture}

            \tikzset{
            voter/.style={circle,draw,minimum size=0.1cm,inner sep=0, fill = black!05, font=\footnotesize},
          }

            \draw [decorate,thick,decoration={brace,amplitude=5pt,mirror,raise=4ex}]
          (0.35,0.3) -- (0.35,-2.3) node[midway,xshift=-5em,yshift=1.5pt]{$(\beta-\varepsilon)x^4$};

            \draw [decorate,thick,decoration={brace,amplitude=5pt,mirror,raise=4ex}]
      (1.1,-0.2) -- (-0.1,-0.2) node[midway,yshift=3em]{$x$};

            \node[voter] (v1) at (0,0) {};
            \node[voter] (v2) at (0.2,0) {};
            \node[voter] (v3) at (0.4,0) {};
            \node at (0.7, 0) {$...$};
            \node[voter] (v4) at (1,0) {};
            \draw[scale = 1] ($(v1)!.5!(v4)$) ellipse (0.6cm and 0.3cm);

            \node[voter] (v1) at (0,-0.8) {};
            \node[voter] (v2) at (0.2,-0.8) {};
            \node[voter] (v3) at (0.4,-0.8) {};
            \node at (0.7, -0.8) {$...$};
            \node[voter] (v4) at (1,-0.8) {};
            \draw[scale = 1] ($(v1)!.5!(v4)$) ellipse (0.6cm and 0.3cm);

            \node at (0.5, -1.4) {$\vdots$};

            \node[voter] (v1) at (0,-2) {};
            \node[voter] (v2) at (0.2,-2) {};
            \node[voter] (v3) at (0.4,-2) {};
            \node at (0.7, -2) {$...$};
            \node[voter] (v4) at (1,-2) {};
            \draw[scale = 1] ($(v1)!.5!(v4)$) ellipse (0.6cm and 0.3cm);

            \draw [decorate,thick,decoration={brace,amplitude=5pt,mirror,raise=4ex}]
          (2.85,-0.3) -- (2.85,-1.7) node[midway,xshift=-3.4em,yshift=1.5pt]{$x^3$};

            \node[voter] (v1) at (2.8,-0.4) {};
            \node[voter] (v2) at (2.8,-0.6) {};
            \node[voter] (v3) at (2.8,-0.8) {};
            \node at (2.8, -1.2) {$\vdots$};
            \node[voter] (v4) at (2.8,-1.6) {};

            \node[ellipse,draw,minimum width=0.5cm,minimum height=1.5cm] (c1) at (2.8,-1) {};
            \node[voter] (v1) at (4, 0) {};
            \node[ellipse,draw,minimum width=0.4cm,minimum height=0.4cm] (c1_) at (4, -0) {};

            \node[ellipse,draw,minimum width=0.75cm,minimum height=1.75cm] (c2) at (2.8,-1) {};
            \node[voter] (v2) at (4, -0.8) {};
            \node[ellipse,draw,minimum width=0.4cm,minimum height=0.4cm] (c2_) at (4, -0.8) {};

            \node at (4, -1.4) {$\vdots$};

            \node[ellipse,draw,minimum width=1.3cm,minimum height=2.2cm] (c3) at (2.95,-1) {};
            \node at (3.35, -1.25) {$...$};
            \node[voter] (v2) at (4, -2) {};
            \node[ellipse,draw,minimum width=0.4cm,minimum height=0.4cm] (c3_) at (4, -2) {};

            \path[draw]
            (c1) edge (c1_)
            (c2) edge (c2_)
            (c3) edge (c3_)
            ;

            \draw [decorate,thick,decoration={brace,amplitude=5pt,mirror,raise=4ex}]
          (3.7,-2.3) -- (3.7,0.3) node[midway,xshift=4.2em,yshift=1.5pt]{$x^4 + 1$};
        \end{tikzpicture}

        \caption{Illustration of the profile constructed in the proofs of \Cref{prop:EJR_vs_Pairs:ub}.
        Block voters are on the left, central voters in the middle, and arm voters on the right.
        Block candidates are approved by $x$ block voters each, whereas arm candidates are approved by all central voters and one arm voter each.}
        \label{fig:prop:EJR_con_app:ub}
    \end{figure}
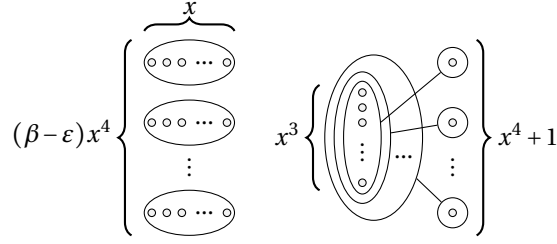

    Fix a committee $W$ that satisfies $\beta$-EJR.
    We claim that for large enough $x$, committee $W$ contains all block candidates.
    To see this, note that there are $(\beta -\varepsilon )x^5$ block voters, $x^3$ central voters, and $x^4+1$ arm voters,
    so $n=(\beta -\varepsilon )x^5+\orderof(x^4)$.
    Thus, $$\frac{n}{\beta k} = \frac{(\beta -\varepsilon )x^5+\orderof(x^4)}{\beta(x^4+1)} = \frac{\beta -\varepsilon} {\beta}\cdot x + \orderof(1)\text.$$
    Hence, we have $\frac{n}{\beta k}<x$ for large enough $x$. Consequently, $\beta$-EJR (with $\ell=1$) demands that a group of $x$ voters who all approve the same candidate has to be represented in $W$, i.e., $W$ must contain a candidate approved by some of these voters.
    In particular, this applies to each block of block voters.
    Since the only way to represent these voters is to include their respective block candidate, $W$ has to contain all block candidates.

    It then follows that $W$
    contains $1 + (1-\beta +\varepsilon) x^4$ arm candidates.
    Consequently, we have
    \begin{align*}
        \Pairs{}(W,\mathcal{E}') &= (\beta - \varepsilon)x^4\binom{x}{2} + \binom{x^3}2 + x^3 \big(1+(1-\beta +\varepsilon)x^4 \big) = (1-\beta +\varepsilon)x^7 + \orderof(x^6)\text.
    \end{align*}
    By contrast, when selecting all arm candidates, we obtain a \Pairs{} score for ${\mathcal E}'$ of $\binom{x^3}2 + x^3(1+x^4) = x^7 + O(x^6)$.
    Considering the ratio of these two values and having $x$ tend to infinity shows that $f$ is
    at most a $(1-\beta+\varepsilon)$-approximation of \Pairs.
    Since $\varepsilon > 0$ with $\varepsilon < \beta$ was chosen arbitrarily, the assertion follows.
\end{proof}

To conclude this section,
we note that the guarantees offered
by \Cref{prop:av_cc_pair_app:lb,prop:pair_ejr_app:lb}
are established by combining separate algorithms for the two objectives in question.
Some of these objectives (in particular, \CC{} and \Pairs{}) do not admit polynomial-time algorithms unless \P$=$\NP.
To achieve a polynomial runtime, we can instead use well-known approximation algorithms for \CC{} and \Pairs{} (for \Pairs{}, we use the sequential Chamberlin--Courant rule on the associated pair instance). This comes, however, at the expense of a factor of up to $1-\frac 1e$ in the approximation ratio~\citep{LacSko-2020-MultiwinnerApproximations}.

\subsection[Cons Objective]{\Connections{} Objective}\label{sec:trade-offs:Connections}

An important reason why we obtained good approximations of \Pairs{}, \AV{}, and \CC{}
was that these objectives are submodular.
In contrast, the \Connections{} objective is not submodular (or even subadditive),
so we cannot use \Cref{prop:submodular}.
In fact, the following result
shows that the trade-off between \Connections{}
and any of \AV{}, \CC{}, or \Pairs{}
is strictly worse (on the side of the \Connections{})
than the trade-offs we have established in \Cref{sec:trade-offs:pairs}.
Notably, our bound
applies even to instances that belong to the VI domain.

\begin{restatable}{proposition}{PairsvsConCovUb}\label{prop:PairsVsConnectionsUb}
    For every $\alpha, \beta \in [0, 1]$,
    if a voting rule satisfies
    $\alpha^2$-\Connections{}
    and $\beta$-\AV, $\beta$-\CC, or $\beta$-\Pairs{},
    then $\alpha + \beta \le 1$.
    This already holds in the VI domain.
\end{restatable}

\begin{proof}
    For the proof of all three statements, we consider the identical instance $(V,A,x^3)$ depicted in \Cref{fig:prop:PairsVsConnectionsUb}.
    The candidate set consists of $x^3$ \emph{block candidates}
    $b_1,\dots,b_{x^3}$
    and $x^3$ \emph{central candidates}
    $c_1,\dots,c_{x^3}$,
    i.e., $2x^3$ candidates in total.
    The voter set $V$ consists of $x^4$ {\em block voters}, split into $x^3$ blocks $B_1, \dots, B_{x^3}$
    of size $x$ each, and $x^3+1$ \emph{central voters},
    i.e., $x^5 + x^3 + 1$ voters in total.
    For $i\in [x^3]$, each voter in block $B_i$ approves candidate $b_i$ only.
    Moreover, the first central voter approves candidate $c_1$, the last central voter approves candidate $c_{x^3}$,
    and for $i=2, \dots, x^3$
    the $i$th central voter approves candidates $c_i$ and $c_{i-1}$.
    We set $k=x^3$.

    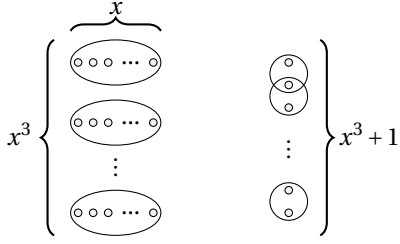
\begin{figure}
        \centering

        \begin{tikzpicture}

            \tikzset{
            voter/.style={circle,draw,minimum size=0.1cm,inner sep=0, fill = black!05, font=\footnotesize},
          }

            \draw [decorate,thick,decoration={brace,amplitude=5pt,mirror,raise=4ex}]
          (0.3,0.3) -- (0.3,-2.3) node[midway,xshift=-3.5em,yshift=1.5pt]{$x^3$};

            \draw [decorate,thick,decoration={brace,amplitude=5pt,mirror,raise=4ex}]
      (1.1,-0.2) -- (-0.1,-0.2) node[midway,yshift=3em]{$x$};

            \node[voter] (v1) at (0,0) {};
            \node[voter] (v2) at (0.2,0) {};
            \node[voter] (v3) at (0.4,0) {};
            \node at (0.7, 0) {$...$};
            \node[voter] (v4) at (1,0) {};
            \draw[scale = 1] ($(v1)!.5!(v4)$) ellipse (0.6cm and 0.3cm);

            \node[voter] (v1) at (0,-0.8) {};
            \node[voter] (v2) at (0.2,-0.8) {};
            \node[voter] (v3) at (0.4,-0.8) {};
            \node at (0.7, -0.8) {$...$};
            \node[voter] (v4) at (1,-0.8) {};
            \draw[scale = 1] ($(v1)!.5!(v4)$) ellipse (0.6cm and 0.3cm);

            \node at (0.5, -1.4) {$\vdots$};

            \node[voter] (v1) at (0,-2) {};
            \node[voter] (v2) at (0.2,-2) {};
            \node[voter] (v3) at (0.4,-2) {};
            \node at (0.7, -2) {$...$};
            \node[voter] (v4) at (1,-2) {};
            \draw[scale = 1] ($(v1)!.5!(v4)$) ellipse (0.6cm and 0.3cm);

            \draw [decorate,thick,decoration={brace,amplitude=5pt,raise=4ex}]
          (2.6,0.3) -- (2.6,-2.3) node[midway,xshift=4.1em,yshift=1.5pt]{$x^3+1$};

    \node[voter] (v1) at (2.8,0) {};
    \node[voter] (v2) at (2.8,-0.3) {};
    \node[voter] (v3) at (2.8,-0.6) {};
    \node at (2.8, -1.15) {$\vdots$};
    \node[voter] (v4) at (2.8,-1.7) {};
    \node[voter] (v5) at (2.8,-2) {};

    \draw[scale=1] ($(v1)!.5!(v2)$) circle (0.25cm);
    \draw[scale=1] ($(v2)!.5!(v3)$) circle (0.25cm);
    \draw[scale=1] ($(v4)!.5!(v5)$) circle (0.25cm);

    \end{tikzpicture}

    \caption{An illustration of the profile constructed in the proof of \Cref{prop:PairsVsConnectionsUb}.
    Block voters are on the left, central voters on the right. Each central candidate covers one pair of voters, but choosing all central candidates yields a large connected component.}
    \label{fig:prop:PairsVsConnectionsUb}
    \end{figure}

    To show that this instance is in VI, we first enumerate the voters in each block, and then the central voters from first to last. By construction, each candidate is approved by an interval of voters.

    The remainder of the proof consists of two parts. The first part shows that, to satisfy $\beta$-\AV, $\beta$-\CC, or $\beta$-\Pairs, we require at least $\beta x^3-\orderof(x^2)$ block candidates. The second part shows that, with the remaining candidates, we can obtain at most a $(1-\beta)^2$ approximation of \Connections.

    \textbf{\AV:} Each block candidate contributes $x$ to the \AV-score, whereas each central candidate only contributes $2$.
    Thus, the maximum \AV-score is $x^4$, achieved by choosing all block candidates.
    If we choose $y\le x^3$ block candidates, then the \AV-score is
    $y x + (x^3 -y)2 $.
    Hence, for a committee of this form to guarantee $\beta$-\AV{}, it has to be the case that
    $y x + (x^3 -y)2 \ge \beta x^4$, or, equivalently,
    $$
    y \ge \frac{\beta x^4 - 2x^3}{x - 2} \ge \frac{\beta x^4 - 2x^3} {x} = \beta x^3 - \orderof(x^2).
    $$

    \textbf{\CC:}
    Each block candidate contributes $x$ to the \CC-score, whereas each central candidate contributes at most~$2$.
    Thus, the maximum \CC-score is $x^4$, achieved by choosing all block candidates.
    If we choose $y\le x^3$ block candidates, then the \CC-score is at most $yx+(x^3-y)2$.
    Hence, for a committee of this form to guarantee $\beta$-\CC{}, it has to be the case that
    $y x + (x^3-y)2\ge \beta x^4$.
    As argued in our analysis for \AV, this implies $y\ge \beta x^3-\orderof(x^2)$.

    \textbf{\Pairs:}
    Each block candidate contributes $\frac{x(x-1)}{2}$ to the \Pairs{}-score, while each central candidate contributes~$1$.
    Thus, the maximum \Pairs-score is $x^3 (\frac{x^2}{2} - \frac{x}{2})$, achieved by choosing all block candidates.
    If we choose $y\le x^3$ block candidates, the \Pairs{} score is at most $x^3 + y (\frac{x^2}{2} - \frac{x}{2})$. Hence, for a committee of this form to guarantee $\beta$-\Pairs{}, it has to be the case that $x^3 + y (\frac{x^2}{2} - \frac{x}{2})\ge \beta x^3 (\frac{x^2}{2} - \frac{x}{2})$ and thus
    $$
    {y} \ge \beta {x^3}- \frac{x^3}{\frac{x^2}{2} - \frac{x}{2}} = \beta {x^3}- \orderof(x).
    $$

    \textbf{\Connections:} We now show that with the remaining $(1-\beta)x^3 + \orderof(x^2)$ candidates, we obtain at best a $(1-\beta)^2$-approximation of \Connections.
    Note that each block candidate contributes $\binom{x}{2}$ to \Connections, hence there are a total of at most $x^3\binom{x}{2}$ connections through block candidates.
    Moreover, $y$ central candidates can connect at most $y+1$ voters.
    Hence, a set of $(1-\beta)x^3 + \orderof(x^2)$ central candidates can cause at most
    $\binom{(1-\beta)x^3 + \orderof(x^2)}{2}$ connections.
    Thus, a committee containing $\beta {x^3}- \orderof(x^2)$ block candidates can achieve a \Connections{} score of at most
    $$x^3\binom{x}{2} + \binom{(1-\beta)x^3 + \orderof(x^2)}{2} = \frac{(1-\beta)^2}{2}\cdot x^6+\orderof(x^5)\text.$$

    By contrast, consider the committee selecting all $x^3$ central candidates.
    This connects all pairs of central voters, and, therefore, achieves a \Connections{} score of $\binom{x^3+1}{2} = \frac {x^6}2 + +\orderof(x^5)$.
    Hence, for $x$ tending to infinity, a committee containing $\beta {x^3}- \orderof(x^2)$ block candidates achieves at most $(1-\beta)^2$-\Connections.
\end{proof}

We derive a similar result for \EJR{} by reducing the number of blocks from $x^3$ to slightly fewer than $\beta x^3$, as in the proof of \Cref{prop:EJR_vs_Pairs:ub}.

\begin{restatable}{proposition}{EJRvsConnectionsTradeoff}\label{prop:EJR_vs_Connections:ub}
    For every $\alpha, \beta \in [0, 1]$,
    if a voting rule satisfies $\alpha^2$-\Connections{}
    and $\beta$-EJR,
    then $\alpha + \beta \le 1$.
    This already holds in the VI domain.
\end{restatable}

\begin{proof}
    Clearly, the statement is true for $\beta = 0$.
    Let $f$ be a voting rule that satisfies $\beta$-EJR for some $\beta \in (0,1]$.
    Fix an $\varepsilon>0$ so that $\varepsilon < \beta$ and $\beta-\varepsilon\in\mathbb Q$.
    We modify the election $\mathcal E$ from the proof of \Cref{prop:PairsVsConnectionsUb} by reducing the number of block candidates (resp., block voters) from $x^3$ to $(\beta-\varepsilon) x^3$ (resp., $(\beta-\varepsilon) x^4$) and only considering values of $x$ for which $(\beta-\varepsilon)x$ is integer
    (as $\beta-\varepsilon$ is rational, there are infinitely many such values).
    The instance is illustrated in \Cref{fig:prop:EJRVsConnectionsUb}.
    As in the proof of \Cref{prop:cc_pair_app:ub}, we want to select $k = x^3$ candidates.
    Denote the resulting election by~${\mathcal E}'$.
    Clearly, ${\mathcal E}'$ remains in VI.

    \begin{figure}
        \centering

        \begin{tikzpicture}

            \tikzset{
            voter/.style={circle,draw,minimum size=0.1cm,inner sep=0, fill = black!05, font=\footnotesize},
          }

            \draw [decorate,thick,decoration={brace,amplitude=5pt,mirror,raise=4ex}]
          (0.3,0.3) -- (0.3,-2.3) node[midway,xshift=-5em,yshift=0.5pt]{$(\beta-\varepsilon) x^3$};

            \draw [decorate,thick,decoration={brace,amplitude=5pt,mirror,raise=4ex}]
      (1.1,-0.2) -- (-0.1,-0.2) node[midway,yshift=3em]{$x$};

            \node[voter] (v1) at (0,0) {};
            \node[voter] (v2) at (0.2,0) {};
            \node[voter] (v3) at (0.4,0) {};
            \node at (0.7, 0) {$...$};
            \node[voter] (v4) at (1,0) {};
            \draw[scale = 1] ($(v1)!.5!(v4)$) ellipse (0.6cm and 0.3cm);

            \node[voter] (v1) at (0,-0.8) {};
            \node[voter] (v2) at (0.2,-0.8) {};
            \node[voter] (v3) at (0.4,-0.8) {};
            \node at (0.7, -0.8) {$...$};
            \node[voter] (v4) at (1,-0.8) {};
            \draw[scale = 1] ($(v1)!.5!(v4)$) ellipse (0.6cm and 0.3cm);

            \node at (0.5, -1.4) {$\vdots$};

            \node[voter] (v1) at (0,-2) {};
            \node[voter] (v2) at (0.2,-2) {};
            \node[voter] (v3) at (0.4,-2) {};
            \node at (0.7, -2) {$...$};
            \node[voter] (v4) at (1,-2) {};
            \draw[scale = 1] ($(v1)!.5!(v4)$) ellipse (0.6cm and 0.3cm);

            \draw [decorate,thick,decoration={brace,amplitude=5pt,raise=4ex}]
          (2.6,0.3) -- (2.6,-2.3) node[midway,xshift=4.1em,yshift=1.5pt]{$x^3+1$};

    \node[voter] (v1) at (2.8,0) {};
    \node[voter] (v2) at (2.8,-0.3) {};
    \node[voter] (v3) at (2.8,-0.6) {};
    \node at (2.8, -1.15) {$\vdots$};
    \node[voter] (v4) at (2.8,-1.7) {};
    \node[voter] (v5) at (2.8,-2) {};

    \draw[scale=1] ($(v1)!.5!(v2)$) circle (0.25cm);
    \draw[scale=1] ($(v2)!.5!(v3)$) circle (0.25cm);
    \draw[scale=1] ($(v4)!.5!(v5)$) circle (0.25cm);

    \end{tikzpicture}

    \caption{An illustration of the profile constructed in the proof of \Cref{prop:EJR_vs_Connections:ub}.
    By reducing the number of blocks in comparison to \Cref{fig:prop:PairsVsConnectionsUb}, each block enforces via $\beta$-EJR that their block candidate is elected.}
    \label{fig:prop:EJRVsConnectionsUb}
    \end{figure}
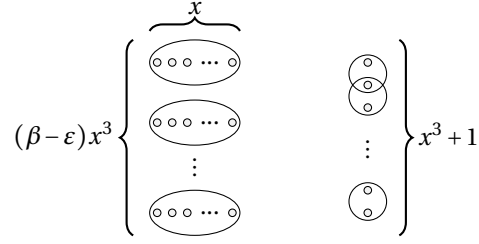

    Fix a committee $W$ that
    satisfies $\beta$-EJR.
    We claim that for large enough $x$, committee $W$ contains all block candidates.
    Each block candidate has a support of $x$ voters and the total number of voters is $n = (\beta-\varepsilon) x^4 + \orderof(x^3)$.
    Hence, for large enough $x$, we have
    $$\frac{n}{\beta k} = \frac{(\beta-\varepsilon) x^4 + \orderof(x^3)}{\beta x^3} = \frac{\beta-\varepsilon}{\beta}\cdot x + \orderof(1)< x\text.$$
    Consequentially, $\beta$-\EJR{} demands that in each block at least one voter approves a candidate of the winning committee.
    Since the voters in each block exclusively approve of the corresponding block candidate, all $(\beta-\varepsilon) x^3$ block candidates have to be contained in $W$.

    It follows that $W$ contains $(1-\beta + \varepsilon)x^3$ central candidates. The highest number of connections caused by central candidates is achieved when they yield one large connected component in the hypergraph induced by the central candidates, i.e., connecting $(1-\beta + \varepsilon)x^3 + 1$ voters.
    Adding this to the connections by block candidates results in

    $$\Connections(W,{\mathcal E}') = (\beta - \varepsilon)x^3\binom{x}{2} + \binom{(1-\beta + \varepsilon)x^3 + 1}{2} = \frac{(1-\beta+\varepsilon)^2 }{2}\cdot x^6 + \orderof(x^5)\text.$$

    By contrast, when selecting all central candidates, we obtain a \Connections{} score of $\frac{x^6}{2} + \orderof(x^5)$.
    Considering the ratio of these two values and having $x$ tend to infinity shows that $f$ is at most a $(1-\beta+\varepsilon)^2$-approximation \Connections. Since $\varepsilon > 0$ with $\varepsilon < \beta$ was chosen arbitrarily, the assertion follows.
\end{proof}

For the \Pairs{} objective, our trade-offs were complemented with straightforward algorithms matching the limitations of the trade-offs (up to rounding).
However, in case of \Connections{}, there is no straightforward method to use an $\alpha$ fraction of candidates to achieve $\alpha^2$-\Connections.
In fact, we will now present a result showing that trade-offs can be even worse
than the ones in \Cref{prop:PairsVsConnectionsUb,prop:EJR_vs_Connections:ub}.

Consider a stepwise function
$s \colon (0,1] \rightarrow [0,1]$
given by $s(\alpha) = 1 / (\lceil 2/\alpha \rceil - 1)$, see \Cref{fig:upper-bounds} for an illustration.
Intuitively, it finds the smallest $p \in \mathbb{N}$
such that $\alpha \ge 2/p$
and returns $1/(p-1)$.
We then have the following trade-off between
\Connections{} and \AV{}.

\begin{proposition}
\label{prop:cons-vs-av:stepwise-ub}
For every $\beta \in [0, 1]$,
    if a voting rule $f$ satisfies $\beta$-\AV, then
     it satisfies at most $s(1-\beta)$-\Connections{}.
\end{proposition}
\begin{proof}
    Let $y = \frac{1}{s(1-\beta)}=\lceil\frac{2}{1-\beta}\rceil -1$.
    Note that $y$ is an integer.
    For an arbitrary constant $x \in \mathbb{N}$,
    consider the election $\mathcal{E} = (V, A, k)$
    defined as follows (see \Cref{fig:prop:cons-vs-av:stepwise-ub} for an illustration).
    The candidate set contains
    $yx^2+1$ \emph{block} candidates
    $b_1,\dots, b_{yx^2+1}$,
    $y$ \emph{arm} candidates
    $a_1,\dots, a_y$, and
    $yx^2$ \emph{chain} candidates
    $(c_{i,j})_{i \in [x^2], j \in [y]}$.
    The voter set $V$ consists of
    $x^2$ \emph{block} voters,
    $yx^3$ \emph{arm} voters split into $y$ arms
    $A_1, \dots, A_y$ of size $x^3$ each,
    $y(x^2-1)$ \emph{chain} voters $(h_{i,j})_{i \in [x^2-1], j \in [y]}$
    split into $y$ arms $H_1, \dots, H_y$ of size $x^2-1$ each,
    and one \emph{central} voter
    $v$.

    The voters have the following preferences.
    All block voters approve all block candidates.
    For each $j\in [y]$, each voter in arm $A_j$ approves the arm candidate $a_j$,
    and additionally, exactly one voter in $A_j$ approves the chain candidate $c_{x^2,j}$.
    For each $i\in [x^2 -1]$ and $j \in [y]$, the chain voter $h_{i, j}$ approves the chain candidates $c_{i,j}$ and $c_{i+1,j}$.
    Finally, the central voter approves the chain candidates $c_{1, j}$ for each $j\in [y]$.

    We set the target committee
    size to $k=y(x^2+1)+1$.
    The high-level idea of the proof is that
    for \Connections{} it is important to connect the arm voters through the selection of chain candidates.
    However, if we select a $\beta$-fraction of block candidates in order to guarantee $\beta$-\AV{},
    then we cannot connect arm voters from different arms.

    Consider a committee $W$ of size $k$.
    Note that $W$ contains at least one block candidate, as there are only $y+yx^2=k-1$ arm and chain candidates.
    Moreover, suppose there is an arm candidate $a$ not included in $W$.
    Then removing a block candidate from $W$ and adding $a$ instead increases both \AV{} and \Connections{}. Thus, we can assume that $W$ contains all $y$ arm candidates.
    For each $z\in\{0,1,\dots, yx^2\}$,
    let $W_z$ be a committee that selects
    $y$ arm candidates, $z+1$ block candidates, and
    $yx^2 - z$ chain candidates.

    Each arm candidate in $W_z$ contributes $x^3$ to the \AV{} score, each block candidate contributes $x^2$, and each chain candidate contributes $2$.
    Thus, we obtain
    \[
        \AV{}(W_z, \mathcal{E}) = yx^3+zx^2 + x^2 + 2(yx^2 - z).
    \]

    Observe that \AV{} is maximized when $z = yx^2$;
    thus, for $W_z$ to provide $\beta$-\AV{},
    it has to be the case that $z \ge \beta yx^2 - \orderof(x)$.

    Now, for \Connections{}, we claim that for large enough $x$, with the remaining $yx^2 - z$ chain candidates, we cannot connect arm voters from two different arms.
    Indeed,
    recall that $y = \frac{1}{s(1-\beta)}=\lceil\frac{2}{1-\beta}\rceil -1<\frac{2}{1-\beta}$.
    Hence, $(1-\beta)y < 2$.
    Consequently, there is an $\varepsilon > 0$
    such that $(1-\beta)y = 2 - \varepsilon$.
    Since $z \ge \beta yx^2 - \orderof(x)$,
    we choose at most $yx^2 - z \le (1-\beta) yx^2 + \orderof(x) < 2x^2 - \varepsilon x^2 + \orderof(x)$
    chain candidates. Thus, for large enough $x$ we have strictly fewer than $2x^2-2$ chain candidates. This proves the claim,
    as it takes $2(x^2-1)$ chain candidates to connect two arm voters from different arms.

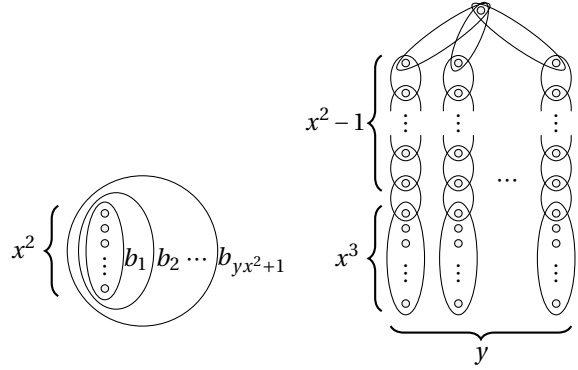
\begin{figure}
    \centering

    \begin{tikzpicture}

        \tikzset{
        voter/.style={circle,draw,minimum size=0.1cm,inner sep=0, fill = black!05, font=\footnotesize},
      }

        \draw [decorate,thick,decoration={brace,amplitude=5pt,mirror,raise=4ex}]
      (0,-0.3) -- (0,-1.5) node[midway,xshift=-3.5em,yshift=1.5pt]{$x^2$};

        \node[voter] (v1) at (0,-0.4) {};
        \node[voter] (v2) at (0,-0.6) {};
        \node[voter] (v3) at (0,-0.8) {};
        \node at (0, -1.1) {$\vdots$};
        \node[voter] (v4) at (0,-1.4) {};

        \node[ellipse,draw,minimum width=0.5cm,minimum height=1.3cm] (c1) at (0,-0.9) {};

        \node[ellipse,draw,minimum width=1cm,minimum height=1.55cm] (c2) at (0.15,-0.9) {};

        \node[ellipse,draw,minimum width=2cm,minimum height=2cm] (c3) at (0.5,-0.9) {};
        \node at (0.41, -1) {$b_1$};
        \node at (1.02, -1) {$b_2 \ \cdots$};
        \node at (1.95, -1.04) {$b_{yx^2+1}$};

        \def\x{4}
        \draw [decorate,thick,decoration={brace,amplitude=5pt,mirror,raise=4ex}]
      (\x + 0.3,-0.3) -- (\x + 0.3,-1.7) node[midway,xshift=-3.5em,yshift=1.5pt]{$x^3$};

        \node[voter] (v_) at (5,2.3) {};

      \draw [decorate,thick,decoration={brace,amplitude=5pt,mirror,raise=4ex}]
      (\x + 0.3,1.7) -- (\x + 0.3,-0.1) node[midway,xshift=-4em,yshift=1.5pt]{$x^2-1$};

        \node[voter] (v1) at (\x,-0.4) {};
        \node[voter] (v2) at (\x,-0.6) {};
        \node[voter] (v3) at (\x,-0.8) {};
        \node at (\x, -1.2) {$\vdots$};
        \node[voter] (v4) at (\x,-1.6) {};

        \node[ellipse,draw,minimum width=0.5cm,minimum height=1.5cm] (c1) at (\x,-1) {};

        \node[ellipse,draw,minimum width=0.4cm,minimum height=0.6cm] (c1) at (\x,-0.2) {};
        \node[voter] (v5) at (\x,0) {};
        \node[ellipse,draw,minimum width=0.4cm,minimum height=0.6cm] (c1) at (\x,0.2) {};
        \node[voter] (v6) at (\x,0.4) {};
        \node[ellipse,draw,minimum width=0.4cm,minimum height=0.6cm] (c1) at (\x,0.6) {};
        \node[voter] (v7) at (\x,0.8) {};
        \node[ellipse,draw,minimum width=0.4cm,minimum height=0.6cm] (c1) at (\x,1) {};
        \node[voter] (v8) at (\x,1.2) {};
        \node[ellipse,white,fill=white,draw,minimum width=0.6cm,minimum height=0.4cm] (c1) at (\x,0.8) {};
        \node[ellipse,draw,minimum width=0.4cm,minimum height=0.6cm] (c1) at (\x,1.4) {};
        \node[voter] (v9) at (\x,1.6) {};
        \node (_) at (\x,0.8) {$\vdots$};
        \node[ellipse,draw,rotate=-55,minimum width=1,minimum height=1.5cm] (c9) at (\x + 0.5, 1.95) {};

        \def\x{4.7}
        \node[voter] (v1) at (\x,-0.4) {};
        \node[voter] (v2) at (\x,-0.6) {};
        \node[voter] (v3) at (\x,-0.8) {};
        \node at (\x, -1.2) {$\vdots$};
        \node[voter] (v4) at (\x,-1.6) {};
        \node[ellipse,draw,minimum width=0.5cm,minimum height=1.5cm] (c1) at (\x,-1) {};

        \node[ellipse,draw,minimum width=0.4cm,minimum height=0.6cm] (c1) at (\x,-0.2) {};
        \node[voter] (v5) at (\x,0) {};
        \node[ellipse,draw,minimum width=0.4cm,minimum height=0.6cm] (c1) at (\x,0.2) {};
        \node[voter] (v6) at (\x,0.4) {};
        \node[ellipse,draw,minimum width=0.4cm,minimum height=0.6cm] (c1) at (\x,0.6) {};
        \node[voter] (v7) at (\x,0.8) {};
        \node[ellipse,draw,minimum width=0.4cm,minimum height=0.6cm] (c1) at (\x,1) {};
        \node[voter] (v8) at (\x,1.2) {};
        \node[ellipse,white,fill=white,draw,minimum width=0.6cm,minimum height=0.4cm] (c1) at (\x,0.8) {};
        \node[ellipse,draw,minimum width=0.4cm,minimum height=0.6cm] (c1) at (\x,1.4) {};
        \node[voter] (v8) at (\x,1.6) {};
        \node (_) at (\x,0.8) {$\vdots$};
        \node[ellipse,draw,rotate=-25,minimum width=0.8,minimum height=1cm] (c9) at (\x + 0.15, 1.95) {};

        \node (_) at (\x + 0.65, 0) {$\cdots$};

        \def\x{6}
        \node[voter] (v1) at (\x,-0.4) {};
        \node[voter] (v2) at (\x,-0.6) {};
        \node[voter] (v3) at (\x,-0.8) {};
        \node at (\x, -1.2) {$\vdots$};
        \node[voter] (v4) at (\x,-1.6) {};
        \node[ellipse,draw,minimum width=0.5cm,minimum height=1.5cm] (c1) at (\x,-1) {};

        \node[ellipse,draw,minimum width=0.4cm,minimum height=0.6cm] (c1) at (\x,-0.2) {};
        \node[voter] (v5) at (\x,0) {};
        \node[ellipse,draw,minimum width=0.4cm,minimum height=0.6cm] (c1) at (\x,0.2) {};
        \node[voter] (v6) at (\x,0.4) {};
        \node[ellipse,draw,minimum width=0.4cm,minimum height=0.6cm] (c1) at (\x,0.6) {};
        \node[voter] (v7) at (\x,0.8) {};
        \node[ellipse,draw,minimum width=0.4cm,minimum height=0.6cm] (c1) at (\x,1) {};
        \node[voter] (v8) at (\x,1.2) {};
        \node[ellipse,white,fill=white,draw,minimum width=0.6cm,minimum height=0.4cm] (c1) at (\x,0.8) {};
        \node[ellipse,draw,minimum width=0.4cm,minimum height=0.6cm] (c1) at (\x,1.4) {};
        \node[voter] (v8) at (\x,1.6) {};
        \node (_) at (\x,0.8) {$\vdots$};
        \node[ellipse,draw,rotate=55,minimum width=1,minimum height=1.5cm] (c9) at (\x - 0.5, 1.95) {};

        \draw [decorate,thick,decoration={brace,amplitude=5pt,mirror,raise=4ex}]
      (3.8,-1.28) -- (6.2,-1.28) node[midway,yshift=-3.3em]{$y$};
    \end{tikzpicture}

    \caption{An illustration of the profile constructed in the proof of \Cref{prop:cons-vs-av:stepwise-ub}.}
    \label{fig:prop:cons-vs-av:stepwise-ub}
\end{figure}

    Let us now calculate the value of the \Connections{} objective.
    The connections among block voters contribute at most $\binom{x^2}{2}$, and the connections among chain voters contribute at most
    $\binom{yx^2}{2}$.
    Connections between chain voters and arm voters contribute at most $yx^5$, as each chain voter can only be connected to arm voters in a single arm.
    We connect all arm voters within the same arm, but, as argued above, we do not connect arm voters from different arms.
    Hence, connections among arm voters contribute $y \binom{x^3}{2}$.
    The central voter can contribute $\orderof(yx^3)$ connections.
    In total,
    \[
        \Connections{}(W_z,\mathcal{E}) =
        \frac{yx^6}{2} + \orderof(x^5)\text.
    \]

    On the other hand, the maximum value of \Connections{} is obtained when we select all $yx^2-y$ chain voters: in this case,
    all $yx^3$ arm voters are connected, and \Connections{} is at least
    \[
       \frac{y^2x^6}{2} + \orderof(x^5)\text.
    \]

    Thus, the fraction of \Connections{} we can obtain while satisfying $\beta$-\AV{} converges to $\frac{1}{y}=s(1-\beta)$ for $x\to\infty$.
\end{proof}

\Cref{prop:cons-vs-av:stepwise-ub} indicates that the trade-offs involving \Connections{} may be rather complex:
for some parameters $\beta$, \Cref{prop:cons-vs-av:stepwise-ub} leads to steeper trade-offs than \Cref{prop:PairsVsConnectionsUb} (see \Cref{fig:upper-bounds}).
This is in contrast to the bound obtained in \Cref{prop:av_pair_app:ub}, where the trade-off between \Pairs{} and \AV{} is tight up to rounding (see \Cref{prop:av_cc_pair_app:lb}).

For example, assume that we want to achieve $\frac 14$-\AV.
Consider an election with a committee size $k$ that is
divisible by~$4$. This means that we can exactly (i.e., without rounding) assign $\frac 14$ of the candidates to achieve $\frac 14$-\AV{} and $\frac 34$ of the candidates to entirely aim at a high \Connections{} score.
If the trade-off in \Cref{prop:PairsVsConnectionsUb} was tight (up to rounding), then we could achieve $\frac 9{16}$-\Connections{}.
However, \Cref{prop:cons-vs-av:stepwise-ub} implies that we can achieve at most $\frac 12$-\Connections{} as $s(1-\frac 14) = \frac 12$.

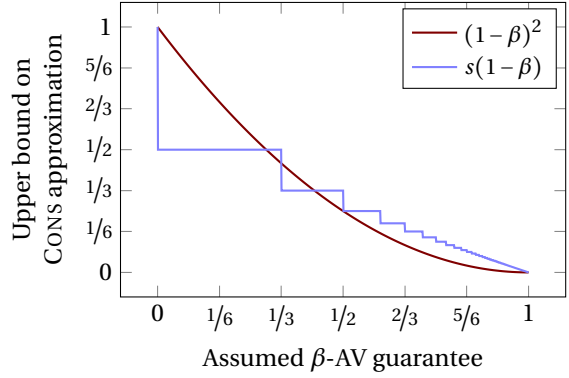
\begin{figure}
    \centering
    \begin{tikzpicture}
    \begin{axis}[
        width=7.5cm,
        height=5.5cm,
        xlabel={Assumed $\beta$-\AV{} guarantee},
        ylabel style={align=center, yshift = 1pt},
        ylabel={Upper bound on \\ \Connections{} approximation},
        xtick={0,0.167,0.333,0.5,0.667,0.833,1},
        xticklabels={$0$, $\nicefrac{1}{6}$, $\nicefrac{1}{3}$, $\nicefrac{1}{2}$, $\nicefrac{2}{3}$, $\nicefrac{5}{6}$, $1$},
        ytick={0,0.167,0.333,0.5,0.667,0.833,1},
        yticklabels={$0$, $\nicefrac{1}{6}$, $\nicefrac{1}{3}$, $\nicefrac{1}{2}$, $\nicefrac{2}{3}$, $\nicefrac{5}{6}$ ,$1$},
        legend style={font=\small},
    ]

    \addplot[color=red!50!black, thick]
    table[x, y index=1] {plotdata.dat};

    \addplot[color=blue!50!white, thick]
    table[x, y index=2] {plotdata.dat};

    \legend{$(1-\beta)^2$, $s(1-\beta)$}

    \end{axis}
    \end{tikzpicture}
    \caption{Two different upper bounds on the possible $\alpha$-approximation of \Connections{} for rules that satisfy $\beta$-\AV{}. The $(1-\beta)^2$ upper bound is the result of \Cref{prop:PairsVsConnectionsUb} and $s(1-\beta)$ is implied by \Cref{prop:cons-vs-av:stepwise-ub}.}
    \label{fig:upper-bounds}
\end{figure}

In fact, our two upper bounds for achieving an $\alpha$-approximation of \Connections{}
given that a voting rule satisfies $\beta$-\AV{}
are based on functions intersecting several times (cf. \Cref{fig:upper-bounds})
and they are of different nature (stepwise versus smooth).
Hence, none of them can yield a tight trade-off.
Consequently, establishing tight trade-offs appears to be a challenging problem.
In particular, because \Connections{} is not submodular,
finding a general lower bound for these trade-offs
seems non-trivial.
Nevertheless, we conclude this section with a positive result on guarantees that we can obtain for a combination of the \Connections{} objective with our other objectives in the VI domain.
It matches a particular intersection point of both of our upper bounds.

\begin{restatable}{proposition}{ConnectionsPositiveForVI}\label{prop:ConnectionsPositiveForVI}
    For every instance $(V, A, k)$ in the VI domain with even $k$, there exists a committee that satisfies $\frac 1 4$-\Connections{} and any one of the criteria $\frac 1 2$-\AV{}, $\frac 1 2$-\CC{}, $\frac 1 2$-\EJR{}, or $\frac 1 2$-\Pairs{}.
\end{restatable}

\begin{proof}
    Let $(V,A,k)$ be an election instance in the VI domain, as certified by the voter ordering $v_1,\dots, v_n$, and let $k$ be even.
    It suffices to show that with $ \frac{k}{2}$ candidates, we can guarantee $\frac{1}{4}$-\Connections{}. With the other $\frac{k}{2}$ candidates, we can use the methods in \Cref{prop:av_cc_pair_app:lb} and \Cref{lemma:equal-shares} to obtain $\frac{1}{2}$-\AV{}, $\frac{1}{2}$-\CC{}, $\frac{1}{2}$-\EJR, or $\frac{1}{2}$-\Pairs{}.

    The proof idea is to fix an optimal solution for \Connections{}, and split it into subsets so that each subset interconnects an interval of voters. An even-sized subset $U$ can then be split into two equal-size contiguous subsets; we choose one that covers at least half of the voters covered by $U$. If some subset contains an odd number of candidates, there must be an even number of such subsets. We show that, given a pair of candidate subsets $U, U'$ of odd size, we can always pick $\left\lfloor\frac{|U|}{2}\right\rfloor$ candidates from $U$, $\left\lfloor\frac{|U'|}{2}\right\rfloor$ candidates from $U'$, and one additional candidate from $U\cup U'$ so as to guarantee $\frac{1}{4}$-approximation to the number of pairs of voters connected by $U$ and $U'$.

    In more detail, let $W$ be a committee that maximizes \Connections{}. Given a candidate $c$,
    let $\ell(c)$ and $r(c)$ denote the index of the leftmost and the rightmost voter approving $c$, respectively.
    Since for \Connections{} it is never desirable to include a pair of candidates $c, c'$ with $V_c\subseteq V_{c'}$,
    we can assume that for all $c, c'\in W$ it holds that $\ell(c)\neq\ell(c')$, and, moreover,
    $\ell(c)<\ell(c')$ implies $r(c)<r(c')$.

    We say that two candidates $c, c'\in W$ {\em overlap} if there is a voter $v\in V_c\cap V_{c'}$.
    Further, we say that a subset $U\subseteq W$ with $U = \{d_1,\dots, d_{t}\}$ forms a {\em chain} if $\ell(d_1)<\dots<\ell(d_t)$ and
    for each $j=1, \dots, t-1$ it holds that $d_j$ and $d_{j+1}$ overlap.
    Note that, by our assumptions, this implies $\ell(d_{j+1}) \le r(d_j) < r(d_{j+1})$
    for all $j\in [t-1]$.
    Given a chain $U$, let $N(U)$ be the set of voters who approve at least one candidate in $U$; by construction, $N(U)$ forms an interval of the voter order and all voters in $N(U)$ are connected by $U$, so a chain $U$ contributes ${{|N(U)|}\choose 2}$ to the \Connections{} objective.

    Observe that $W$ can be partitioned into pairwise disjoint inclusion-maximal chains $U_1, \dots, U_\mu$, so that the candidates from different chains do not overlap. Moreover, since $k$ is even, the number of odd-sized chains in this partition is even. Divide the set of all odd-sized chains into pairs.

    Consider an even-sized chain $U = \{d_1,\dots, d_{2s}\}$ in this partition. Assume that $|N_U|=x$, i.e, $U$ connects $\binom{x}{2}$ pairs of voters.
    Let $z=\lfloor\frac{x}{2}\rfloor+1$. Consider the voter $r(d_s)$.
    If $r(d_s) -\ell(d_1) +1 > \frac{x}{2}$, then candidates $\{d_1,\dots,d_s\}$ form a chain that covers at least $z$ voters and hence contributes at least $\binom{z}{2}$ to the \Connections{} objective. Else, $\ell(d_{s+1}) \le r(d_s)\le \ell(d_1)+\frac{x}{2}-1$ and hence candidates $\{d_{s+1},\dots, d_t\}$ form a chain that covers at least $z$ voters and contributes at least $\binom{z}{2}$ to the \Connections{} objective. It remains to note that $\binom{z}{2}\ge \frac 14\cdot\binom{x}{2}$.

    On the other hand, suppose that $(U, U')$ is a pair of odd-sized chains from our partition of $W$,
    with $U=\{d_1,\dots, d_{2s+1}\}$,
    $U'=\{e_1,\dots, e_{2s'+1}\}$, and $|N(U)|=x, |N(U')|=y$.
    We will consider two ways of selecting $\frac{|U|+|U'|}{2}$ candidates from $U\cup U'$: (1) choose $s+1$ candidates from $U$ and $s'$ candidates from $U'$ or (2) choose $s$ candidates from $U$ and $s'+1$ candidates from $U'$. We will argue that at least one of these options contributes at least $\frac14$ of the connections provided by $U\cup U'$.

    We will use the following technical lemma.

    \begin{lemma}\label{lem:binom2}
    For every $a, b\in \mathbb N$ with $a\le b$ we have
    $$
    \binom{a}{2}+\binom{b-a+1}{2}\ge \frac12\cdot\binom{b}{2}.
    $$
    \end{lemma}
    \begin{proof}
    Since $b\in\mathbb N$, we have $1-b\le 0$ and hence
    $(2a-b+1)^2\ge 1-b$. Expanding the left-hand side, we obtain
    $$
    4a^2+b^2+1-4ab-2b+4a\ge 1-b.
    $$
    Consequently, we have
    \begin{align*}
    4\cdot\left(\binom{a}{2}+\binom{b-a+1}{2}\right) &= 2a(a-1)+2(b-a)(b-a+1) =
    2a^2-2a+2b^2-2ab-2ab+2a^2+2b-2a \\
    &=(4a^2+b^2+1-4ab-2b-4a)+b^2+4b-1
    \ge (1-b) + b^2+4b-1 \ge b(b-1) =2\cdot\binom{b}{2}.
    \end{align*}
    Dividing both sides by 4, we get the desired result.
    \end{proof}

    Suppose the chain $U_1 = \{d_1,\dots, d_{s+1}\}$ covers $x_1$ voters in $N(U)$. Since $\ell(d_{s+2})\le r(d_{s+1})$, the chain $U_2 = \{d_{s+2},\dots, d_t\}$ covers at least $x-x_1+1$ voters in $N(U)$.
    Similarly, if the chain $U'_1 = \{e_1,\dots, e_{s'+1}\}$ covers
    $y_1$ voters in $N(U')$, then the chain $U'_2 = \{e_{s'+2},\dots, e_{2s'+1}\}$ covers at least $y-y_1+1$ voters in $N(U')$.
    Consequently, the number of connections provided by $U_1\cup U'_2$ is at least $z = \binom{x_1}{2}+\binom{y-y_1+1}{2}$, whereas the number of connections provided by $U'_1\cup U_2$ is at least $z' = \binom{y_1}{2}+\binom{x-x_1+1}{2}$.
    Note that
    $$
    z+z' = \binom{x_1}{2}+\binom{x-x_1+1}{2}+\binom{y-y_1+1}{2}+\binom{y_1}{2} \ge \frac12\cdot\left(\binom{x}{2}+\binom{y}{2}\right),
    $$
    where the inequality follows by applying \Cref{lem:binom2} twice, first with $b=x, a=x_1$ and second with $b=y, a=y_1$.
    Hence, $\max\{z, z'\}\ge \frac14\cdot (\binom{x}{2}+\binom{y}{2})$, i.e.,
    at least one of the sets $U_1\cup U'_2$ and $U'_1\cup U_2$ provides at least $\frac14$ of all connections provided by $U\cup U'$.

    To summarize, we argued that we can select half of the candidates from each even-sized chain and half of the candidates from each pair of odd-sized chains so as to obtain at least $\frac14$ of the connections provided by the respective chain/pair of chains. This concludes the proof.
\end{proof}

For $\frac 14$-\Connections{} and $\frac 12$-EJR, the result of \Cref{prop:ConnectionsPositiveForVI} extends to odd $k$.
Indeed, it suffices to use $\lfloor \frac k2\rfloor$ candidates to achieve $\frac 12$-EJR (\Cref{lemma:equal-shares}).
Moreover, for odd $k$, the proof of \Cref{prop:ConnectionsPositiveForVI} implies that choosing $\lceil \frac k2\rceil$ candidates yields $\frac 14$-\Connections{} in the modified election where we are allowed to select $k+1$ instead of $k$ candidates, i.e., we achieve an approximation of a value that can only be larger than the maximum \Connections{} score in the original election.

\section{Conclusion}

Our paper sheds new light on the interdependency of mass and elite polarization.
We observe that the selection of a representative committee can significantly influence elite polarization independently of mass polarization.
With the aim of avoiding polarization at the level of the representation, we have introduced \Pairs{} and \Connections{}, two numerical objectives that measure how well a committee interlaces the electorate.

We show that, while maximizing both objectives is \NP-hard, a committee maximizing either of them can be computed in polynomial time on the voter-candidate interval domain.
Also, we study the compatibility of our objectives with measures of excellence, diversity, and proportionality.
We identify approximation trade-offs suggesting that, in the worst case, one cannot improve over the simple strategy of dividing the committee seats among different objectives and maximizing each objective with its designated share of the seats: there are instances on which the synergies are negligible.
For almost all objectives we study, a subcommittee yields a fraction of the optimal value that is proportional to its size. Only for \Connections{}, the dependency is quadratic (or even worse), leading to inferior guarantees.

We believe that our work offers an important perspective that has been missing from the social choice literature on multiwinner voting.
As such, it calls for further research; in what follows, we suggest some promising directions.

For future work, an immediate open question is to determine the exact trade-off between \Connections{} and other objectives.
While we have a bound for $\alpha^2$-\Connections{} and $(1-\alpha)$-approximations of other objectives, \Cref{prop:cons-vs-av:stepwise-ub} shows that the picture is more nuanced.

Going beyond our base model,
another direction is to
consider our objectives in the broader context of participatory budgeting (PB), where each candidate has a cost, and the committee needs to stay within a given budget \citep{ReMa23a}.
In this setting, candidates are usually projects, such as a playground, a community garden, or a cycling path.
Interlacing voters by projects in PB has an additional interpretation: the funded projects may lead to interaction among the agents who use them
(e.g., working together in a community garden).
This seems quite desirable in the context of PB, where one of the goals is community building.

Moreover, it would be interesting to explore the compatibility of our objectives and the canonical desiderata in the context of real-life instances: It is plausible that on realistic data one can achieve much better trade-offs than in the worst case.

Finally, while \Pairs{} and \Connections{} offer some insight into the polarization induced by a committee, there are settings where they fail to provide useful information:
For example, if some candidate is approved by all voters, any committee containing this candidate maximizes both objectives.
Therefore, further insights could be gained by studying refined versions of our objectives; e.g., one can consider the strength of the connections or,
in case of the \Connections{} objective,
the length of the (shortest) path between a pair of voters.

\section*{Acknowledgements}
Most of this work was done when Edith Elkind was at the University of Oxford.
Chris Dong was supported by the Deutsche Forschungsgemeinschaft under grants BR 2312/11-2 and BR 2312/12-1.
Martin Bullinger was supported by the AI Programme of The Alan Turing Institute.
Tomasz W\k{a}s and Edith Elkind were supported by the UK Engineering and Physical Sciences Research Council (EPSRC) under grant EP/X038548/1.

\bibliography{abb,interlacing}

\begin{thebibliography}{37}
\providecommand{\natexlab}[1]{#1}
\providecommand{\url}[1]{\texttt{#1}}
\expandafter\ifx\csname urlstyle\endcsname\relax
  \providecommand{\doi}[1]{doi: #1}\else
  \providecommand{\doi}{doi: \begingroup \urlstyle{rm}\Url}\fi

\bibitem[Abramowitz and Saunders(2008)]{AbSa08a}
Alan~I. Abramowitz and Kyle~L. Saunders.
\newblock Is polarization a myth?
\newblock \emph{The Journal of Politics}, 70\penalty0 (2):\penalty0 542--555, 2008.

\bibitem[Aziz et~al.(2017)Aziz, Brill, Conitzer, Elkind, Freeman, and Walsh]{AziBriConElkFreEtal-2017-EJR}
Haris Aziz, Markus Brill, Vincent Conitzer, Edith Elkind, Rupert Freeman, and Toby Walsh.
\newblock Justified representation in approval-based committee voting.
\newblock \emph{Social Choice and Welfare}, 48\penalty0 (2):\penalty0 461--485, 2017.

\bibitem[Barber{\`a} and Coelho(2008)]{BarCoe-2008-CommitteeMonotonicity}
Salvador Barber{\`a} and Danilo Coelho.
\newblock How to choose a non-controversial list with $k$ names.
\newblock \emph{Social Choice and Welfare}, 31\penalty0 (1):\penalty0 79--96, 2008.

\bibitem[Betzler et~al.(2013)Betzler, Slinko, and Uhlmann]{BSU13}
Nadja Betzler, Arkadii Slinko, and Johannes Uhlmann.
\newblock On the computation of fully proportional representation.
\newblock \emph{Journal of Artificial Intelligence Research}, 47:\penalty0 475--519, 2013.

\bibitem[Bol et~al.(2019)Bol, Matakos, Troumpounis, and Xefteris]{BolMatTroXef-2019-VotingEquilibria}
Damien Bol, Konstantinos Matakos, Orestis Troumpounis, and Dimitrios Xefteris.
\newblock Electoral rules, strategic entry and polarization.
\newblock \emph{Journal of Public Economics}, 178:\penalty0 104065, 2019.

\bibitem[Brill and Peters(2023)]{BrPe23a}
Markus Brill and Jannik Peters.
\newblock Robust and verifiable proportionality axioms for multiwinner voting.
\newblock In \emph{Proceedings of the 24th ACM Conference on Economics and Computation (ACM-EC)}, 2023.

\bibitem[Brill and Peters(2024)]{BriPet-2024-MultigoalMultiwinner}
Markus Brill and Jannik Peters.
\newblock Completing priceable committees: Utilitarian and representation guarantees for proportional multiwinner voting.
\newblock In \emph{Proceedings of the 38th AAAI Conference on Artificial Intelligence}, pages 9528--9536, 2024.

\bibitem[Chamberlin and Courant(1983)]{ChaCou-1983-ChamberlinCourant}
John~R. Chamberlin and Paul~N. Courant.
\newblock Representative deliberations and representative decisions: {P}roportional representation and the {B}orda rule.
\newblock \emph{American Political Science Review}, 77\penalty0 (3):\penalty0 718--733, 1983.

\bibitem[Colley et~al.(2023)Colley, Grandi, Hidalgo, Macedo, and Navarrete]{ColGraHidMacNav-2023-Polarization}
Rachael Colley, Umberto Grandi, C{\'e}sar Hidalgo, Mariana Macedo, and Carlos Navarrete.
\newblock Measuring and controlling divisiveness in rank aggregation.
\newblock In \emph{Proceedings of the 32nd International Joint Conference on Artificial Intelligence (IJCAI)}, pages 2616--2623, 2023.

\bibitem[Cox(1985)]{Cox-1985-EquilibriumApproval}
Gary~W. Cox.
\newblock Electoral equilibrium under approval voting.
\newblock \emph{American Journal of Political Science}, pages 112--118, 1985.

\bibitem[Delemazure et~al.(2024)Delemazure, Janeczko, Kaczmarczyk, and Szufa]{DelJanKacSzu-2024-ConflictingPair}
Th{\'e}o Delemazure, {\L}ukasz Janeczko, Andrzej Kaczmarczyk, and Stanis{\l}aw Szufa.
\newblock Selecting the most conflicting pair of candidates.
\newblock In \emph{Proceedings of the 33rd International Joint Conference on Artificial Intelligence (IJCAI)}, pages 2766--2773, 2024.

\bibitem[Do et~al.(2022)Do, Hervouin, Lang, and Skowron]{DHLS22a}
Virginie Do, Matthieu Hervouin, J{\'e}r{\^o}me Lang, and Piotr Skowron.
\newblock Online approval committee elections.
\newblock In \emph{Proceedings of the 31st International Joint Conference on Artificial Intelligence (IJCAI)}, pages 251--257, 2022.

\bibitem[Dong and Peters(2025)]{DoPe25a}
Chris Dong and Jannik Peters.
\newblock Proportional multiwinner voting with dynamic candidate sets.
\newblock In \emph{Proceedings of the 42nd International Conference on Machine Learning}, 2025.
\newblock Forthcoming.

\bibitem[Elkind and Ismaili(2015)]{ElkIsm-2015-CCforOWA}
Edith Elkind and Anisse Ismaili.
\newblock {OWA}-based extensions of the {C}hamberlin--{C}ourant rule.
\newblock In \emph{Proceedings of the 4th International Conference on Algorithmic Decision Theory (ADT)}, pages 486--502, 2015.

\bibitem[Elkind and Lackner(2015)]{EL15}
Edith Elkind and Martin Lackner.
\newblock Structure in dichotomous preferences.
\newblock In \emph{Proceedings of the 24th International Joint Conference on Artificial Intelligence (IJCAI)}, pages 2019--2025, 2015.

\bibitem[Elkind et~al.(2017)Elkind, Faliszewski, Skowron, and Slinko]{ElkFalSkoSli-2017-Multiwinner}
Edith Elkind, Piotr Faliszewski, Piotr Skowron, and Arkadii Slinko.
\newblock Properties of multiwinner voting rules.
\newblock \emph{Social Choice and Welfare}, 48:\penalty0 599--632, 2017.

\bibitem[Elkind et~al.(2024)Elkind, Faliszewski, Igarashi, Manurangsi, Schmidt-Kraepelin, and Suksompong]{ElkFalIgaManSchEtal-2024-PriceOfJR}
Edith Elkind, Piotr Faliszewski, Ayumi Igarashi, Pasin Manurangsi, Ulrike Schmidt-Kraepelin, and Warut Suksompong.
\newblock The price of justified representation.
\newblock \emph{ACM Transactions on Economics and Computation}, 12\penalty0 (3):\penalty0 1--27, 2024.

\bibitem[Fairstein et~al.(2022)Fairstein, Vilenchik, Meir, and Gal]{FaiVilMeiGal-2022-MultigoalPB}
Roy Fairstein, Dan Vilenchik, Reshef Meir, and Kobi Gal.
\newblock Welfare vs. representation in participatory budgeting.
\newblock In \emph{Proceedings of the 21st International Conference on Autonomous Agents and Multiagent Systems (AAMAS)}, pages 409--417, 2022.

\bibitem[Faliszewski et~al.(2017)Faliszewski, Skowron, Slinko, and Talmon]{FSST17a}
Piotr Faliszewski, Piotr Skowron, Arkadii Slinko, and Nimrod Talmon.
\newblock Multiwinner voting: A new challenge for social choice theory.
\newblock In Ulle Endriss, editor, \emph{Trends in Computational Social Choice}, chapter~2. 2017.

\bibitem[Fiorina(2017)]{Fior17a}
Morris~P. Fiorina.
\newblock \emph{Unstable majorities: Polarization, party sorting, and political stalemate}.
\newblock Hoover press, 2017.

\bibitem[Fiorina and Abrams(2008)]{FiAb08a}
Morris~P. Fiorina and Samuel~J. Abrams.
\newblock Political polarization in the {A}merican public.
\newblock \emph{Annual Review of Political Science}, 11\penalty0 (1):\penalty0 563--588, 2008.

\bibitem[Fiorina et~al.(2011)Fiorina, Abrams, and Pope]{FAP11a}
Morris~P. Fiorina, Samuel~J. Abrams, and Jeremy~C. Pope.
\newblock \emph{Culture war? The myth of a polarized {A}merica}.
\newblock Longman, 3rd edition, 2011.

\bibitem[Garey and Johnson(1979)]{GaJo79a}
Michael~R. Garey and David~S. Johnson.
\newblock \emph{Computers and Intractability: A Guide to the Theory of NP-Completeness}.
\newblock W. H. Freeman, 1979.

\bibitem[Godziszewski et~al.(2021)Godziszewski, Batko, Skowron, and Faliszewski]{GodBatSkoFal-2021-2DApprovals}
Micha{\l}~T. Godziszewski, Pawe{\l} Batko, Piotr Skowron, and Piotr Faliszewski.
\newblock An analysis of approval-based committee rules for 2{D}-{E}uclidean elections.
\newblock In \emph{Proceedings of the 35th AAAI Conference on Artificial Intelligence}, pages 5448--5455, 2021.

\bibitem[Kocot et~al.(2019)Kocot, Kolonko, Elkind, Faliszewski, and Talmon]{KocKolElkFalTal-2019-MultigoalMultiwinner}
Maciej Kocot, Anna Kolonko, Edith Elkind, Piotr Faliszewski, and Nimrod Talmon.
\newblock Multigoal committee selection.
\newblock In \emph{Proceedings of the 28th International Joint Conferences on Artificial Intelligence (IJCAI)}, pages 385--391, 2019.

\bibitem[Kurella and Barbaro(2024)]{KurBar-2024-PolarizingProp}
Anna-Sophie Kurella and Salvatore Barbaro.
\newblock On the polarization premium for radical parties in {PR} electoral systems.
\newblock Technical report, Gutenberg School of Management and Economics, Johannes Gutenberg-Universit{\"a}t Mainz, 2024.

\bibitem[Lackner and Skowron(2020)]{LacSko-2020-MultiwinnerApproximations}
Martin Lackner and Piotr Skowron.
\newblock Utilitarian welfare and representation guarantees of approval-based multiwinner rules.
\newblock \emph{Artificial Intelligence}, 288:\penalty0 103366, 2020.

\bibitem[Lackner and Skowron(2023)]{LaSk22b}
Martin Lackner and Piotr Skowron.
\newblock \emph{Multi-Winner Voting with Approval Preferences}.
\newblock Springer Nature, 2023.

\bibitem[Levendusky and Malhotra(2016)]{LeMa16a}
Matthew Levendusky and Neil Malhotra.
\newblock Does media coverage of partisan polarization affect political attitudes?
\newblock \emph{Political Communication}, 33\penalty0 (2):\penalty0 283--301, 2016.

\bibitem[Levin et~al.(2021)Levin, Milner, and Perrings]{levin2021dynamics}
Simon~A Levin, Helen~V Milner, and Charles Perrings.
\newblock The dynamics of political polarization.
\newblock \emph{Proceedings of the National Academy of Sciences}, 118\penalty0 (50):\penalty0 e2116950118, 2021.

\bibitem[Monroe(1995)]{Mon-1995-Monroe}
Burt~L. Monroe.
\newblock Fully proportional representation.
\newblock \emph{American Political Science Review}, 89\penalty0 (4):\penalty0 925--940, 1995.

\bibitem[Myerson and Weber(1993)]{MyeWeb-1993-VotingEquilibria}
Roger~B. Myerson and Robert~J. Weber.
\newblock A theory of voting equilibria.
\newblock \emph{American Political Science Review}, 87\penalty0 (1):\penalty0 102--114, 1993.

\bibitem[Nemhauser et~al.(1978)Nemhauser, Wolsey, and Fisher]{nemhauser1978analysis}
George~L Nemhauser, Laurence~A Wolsey, and Marshall~L Fisher.
\newblock An analysis of approximations for maximizing submodular set functions—{I}.
\newblock \emph{Mathematical programming}, 14:\penalty0 265--294, 1978.

\bibitem[Peters and Skowron(2020)]{PeSk20a}
Dominik Peters and Piotr Skowron.
\newblock Proportionality and the limits of welfarism.
\newblock In \emph{Proceedings of the 21nd ACM Conference on Economics and Computation (ACM-EC)}, pages 793--794, 2020.

\bibitem[Phragm{\'e}n(1899)]{Phra99a}
Edvard Phragm{\'e}n.
\newblock Till fr{\aa}gan om en proportionell valmetod.
\newblock \emph{Statsvetenskaplig Tidskrift}, 2\penalty0 (2):\penalty0 297--305, 1899.

\bibitem[Rey and Maly(2023)]{ReMa23a}
Simon Rey and Jan Maly.
\newblock The (computational) social choice take on indivisible participatory budgeting.
\newblock Technical report, https://arxiv.org/pdf/2303.00621.pdf, 2023.

\bibitem[Thiele(1895)]{Thie95a}
Thorvald~N. Thiele.
\newblock Om flerfoldsvalg.
\newblock \emph{Oversigt over det Kongelige Danske Videnskabernes Selskabs Forhandlinger}, pages 415--441, 1895.

\end{thebibliography}

\clearpage

\appendix

\section*{Appendix}

In the appendix, we present proofs missing from the main part of the paper.
\section{Implications of Submodularity}

In this section, we prove our proposition concerning submodular functions.

\submod*

\begin{proof}
    Let $S=\{x_1, \dots, x_k\}$. We can write $f(S)$ as a telescoping sum
    $$
    f(S) = \sigma_1+\dots+\sigma_k,
    $$
    where $\sigma_i=f(\{x_1, \dots, x_i\})-f(\{x_1, \dots, x_{i-1}\})$ for each $i\in [k]$.
    Let $\sigma_{i_1}, \dots, \sigma_{i_\ell}$, $i_1\le\dots\le i_\ell$, be the $\ell$ largest
    summands in this sum, and set $S'=\{x_{i_1}, \dots, x_{i_\ell}\}$.
    Note that $\sum_{j\in [\ell]}\sigma_{i_j}\ge \frac{\ell}{k}\cdot f(S)$.
    Moreover, we can write $f(S')$ as a telescoping sum
    $$
    f(S')=\sigma'_{i_1}+\dots+\sigma'_{i_\ell},
    $$
    where $\sigma'_{i_j}=f(x_{i_1}, \dots, x_{i_j})-f(x_{i_1}, \dots, x_{i_{j-1}})$ for each $j\in [\ell]$.
    It remains to note that
    $\{x_{i_1}, \dots, x_{i_{j-1}}\}\subseteq \{x_1, x_2, \dots, x_{i_{j-1}}\}$ for all $j\in [\ell]$ and hence by submodularity we have
    $\sigma'_j\ge \sigma_j$ for all $j\in[\ell]$.
    Therefore,
    $$
    f(S')=\sigma'_{i_1}+\dots\sigma'_{i_\ell}\ge \sigma_{i_1}+\dots\sigma_{i_\ell} \ge \frac{\ell}{k}\cdot f(S),
    $$
    which is what we wanted to prove.
\end{proof}

\section[Cons in VCI Domain]{Efficient Maximization of the \Connections{} Objective in the VCI Domain}
\label{app:DP}
In this section, we provide a full specification of the dynamic program for maximizing \Connections{} in the VCI domain.
We then prove \Cref{thm:vci-conn}, restated as follows.

\VCIconn*

\begin{proof}
    Consider an election ${\mathcal E}=(V, A, k)$.
    By \Cref{prop:VCI-to-CI}, we may assume without loss of generality that $\mathcal E$ belongs to the CI domain.
    Moreover, by polynomial-time preprocessing, we obtain a candidate order $c_1, \dots, c_m$ witnessing membership in the CI domain \citep{EL15}.
    For each voter $v\in V$, let $\ell(v):=\min\{i\colon c_i\in A_v\}$ and $r(v):=\max\{i\colon c_i\in A_v\}$ be the leftmost and the rightmost approved candidates of voter $v$, respectively.
    Given candidate indices $1\le j < i\le m$, let $V(\lnot j, i)$ be the set of voters that approve $c_i$, but not~$c_j$,
    and let $n(\lnot j, i)$ be the size of this set.
    Formally, we define
    $$
    V(\lnot j,i) := \{v\in V\colon j < \ell(v) \le i\le r(v)\} \quad\text{ and }\quad n(\lnot j,i):=\lvert V(\lnot j,i)\rvert.
    $$
    Further, we introduce an indicator variable $\mathbb{1}(j\land i)$ defined as
    \begin{equation*}
        \mathbb{1}(j\land i) :=
        \begin{cases}
            \text{true} &\text{if } V_{c_j}\cap V_{c_i} \neq \emptyset\\
            \text{false} &\text{otherwise.}
        \end{cases}
    \end{equation*}
    That is, $\mathbb{1}(j\land i)$ is true if and only if there is a voter that approves both $c_j$ and $c_i$.
    \paragraph{Calculating the number of connected pairs after adding a candidate.}
    Consider adding $c_i$ to a committee $W\subseteq \{c_1, \dots, c_{i-1}\}$, i.e., $c_i$ is to the right of all candidates in $W$ with respect to the candidate order.
    Let $j^*\in [i-1]$ be the index of the rightmost candidate in $W$.
    We will now discuss how to update $\conpair(W\cup \{c_i\})$. The update procedure depends
    on the value of $\mathbb{1}(j^*\land i)$, i.e, on whether there is a voter that approves both $c_{j^*}$ and $c_i$.

    First, suppose that $W=\emptyset$ or $\mathbb{1} (j^*\land i)$ is false.
    Then we claim that no voter in $V_{c_i}$ approves any candidate in $W$. This is obvious if $W=\emptyset$.
    On the other hand, if $\mathbb{1} (j^*\land i)$ is false, suppose for contradiction that some voter $v\in V_{c_i}$ approves some candidate $c_j\in W$. Then by the choice of $j^*$ we have $j<j^*$, and in the CI domain $c_j,c_i\in A_v$ implies $c_{j^*}\in A_v$, a contradiction.

    Thus, adding $c_i$ interconnects the voters approving $c_i$,
    but does not connect any of them to any other voters.
    Hence, in both cases,
    the number of connected pairs after adding $c_i$ to $W$ is
    $$
    \conpair(W\cup \{c_i\}) = \conpair(W) +
    \binom{\lvert V_{c_i}\rvert}{2}.
    $$

    Now, suppose that $\mathbb{1} (j^*\land i)$ is true. By CI, if a voter approves both $c_i$ and some $c_j\in W$, then they also approve $c_{j^*}$.
    The update now depends on the connected component containing $c_{j^*}$ in the hypergraph induced by $W\cup \{c_i\}$.
    Given a set of candidates $W$ and a candidate $c\in W$,
    we define
    $\concomp{c}{W} := \{u\in V\colon u\sim_W v ~\text{for some $v\in V_c$}\}$.
    Note that $\concomp{c}{W}$ is well-defined, because $\sim_W$ is an equivalence relation and $v\sim_W v'$ for all $v, v' \in V_c$.
    In fact, $\concomp{c}{W}$ is exactly the set of voters in the connected component containing the hyperedge~$c$.

    Now, adding $c_i$ creates two types of connections:
    those among the newly connected voters in $V(\lnot j^*, i)$ and those between $V(\lnot j^*, i)$ and the voters in the connected component of $c_{j^*}$, i.e., $\concomp{c_{j^*}}{W}$. Therefore,
    we have
    $$
    \concomp{c_i}{W\cup \{c_i\}} = V(\lnot j^*, i)\cupdot \concomp{c_{j^*}}{W},
    $$
    where $\cupdot$ denotes the disjoint union of two sets. Consequently,
    $$
    \conpair(W\cup \{c_i\}) = \conpair(W) + \binom{n(\lnot j^*,i)}{2} + n(\lnot j^*, i) \cdot \lvert \concomp{c_{j^*}}{W}\rvert.
    $$

    \paragraph{Defining the dynamic program.}
    For each $i\in [m]$, $b\in [k]$, $x\in \{0\}\cup[n]$,
    let $\opt[i,x,b]$ denote the maximum number of voter pairs that can be connected by a committee of size at most $b$ that has $c_i$ as its rightmost candidate, with $c_i$ being in a connected component that contains $x$ voters.
    We use the convention that $\opt[i,x,b] = -1$ if there is no such committee.
    We will define functions $\dyn[i,x,b]$ and $W[i,x,b]$ and argue that for all $i,x,b$ it holds that $\dyn[i,x,b] = \opt[i,x,b]$ and, moreover, if this value is nonnegative, $W[i,x,b]$ is a committee of size at most $b$ with rightmost candidate $c_i$ being in a connected component that contains $x$ voters, which satisfies $\Connections(W[i, x, b])=\opt[i, x, b]$.

    We initialize the dynamic program by handling the case $b=1$. For convenience, we also deal with the case $i=1$ at this point.
    \begin{itemize}
        \item
        For each $i\in [m]$,
        the number of pairs connected by $\{c_i\}$ is $\binom{|V_{c_i}|}{2}$.
        Hence, we set $\dyn[i,\lvert V_{c_i}\rvert,1] = \binom{|V_{c_i}|}{2}$ and
        $W[i,\lvert V_{c_i}\rvert, 1]=\{c_i\}$.
        For $x\neq \lvert V_{c_i}\rvert$, we set
        $\dyn[i,x,1] = -1$.
        \item
        For each $b\in [k]$, we set $\dyn[1,\lvert V_{c_1}\rvert,b] = \binom{|V_{c_1}|}{2}$:
        this is the number of pairs connected by $\{c_1\}$.
        Also, set $W[1,\lvert V_{c_1}\rvert, b]=\{c_1\}$, and $\dyn[1,x,b] = -1$ for $x\neq \lvert V_{c_1}\rvert$.
    \end{itemize}
    Clearly,
    it holds that $\dyn[i,x,b] = \opt[i,x,b]$ for $b=1$ and all $i\in [m]$, $x\in\{0\}\cup[n]$ as well as for $i=1$ and all $x\in\{0\}\cup [n]$, $b\in[k]$. Moreover, for triples $(i, x, b)$ that satisfy $b=1$ or $i=1$
    it holds that
    if $\dyn[i,x,b]\ge 0$, then $W[i,x,b]$ is a committee that provides $\opt[i,x,b]$ connections.

    Having computed $\dyn[i,x, 1]$ and $W[i,x, 1]$ for all $i\in [m]$, $x\in\{0\}\cup[n]$, we proceed in $k-1$ stages: in stage $t\in [k-1]$, we compute $\dyn[i,x, t+1]$ and $W[i,x, t+1]$ for all $i\in [m]$, $x\in\{0\}\cup[n]$.
    Within each stage, we proceed in increasing order of $i$ (note that the case $i=1$ was handled during the initialization stage).
    Thus, it remains to explain how to fill out the cells $\dyn[i, x, b]$ and $W[i, x, b]$ with $i,b\ge 2$, given that we have already filled out $\dyn[j,y,b-1]$ and $W[j, y,b-1]$ for all $j<i$ and all $y\in\{0\}\cup[n]$.

    For a given triple $(i,x, b)$ and each $j\in [i-1]$, $y\in\{0\}\cup[n]$, we
    define $\score(i,x,b,j,y)$ as follows:
    \begin{itemize}
        \item If $\mathbb{1}(j\land i)$ is true, $y+n(\lnot j,i) = x$, and $\dyn[j,y,b-1]\ge 0$, set $\score(i,x,b,j,y) = \dyn[j,y,b-1] + \binom{x-y}{2} +y(x-y)$.
        \item If $\mathbb{1}(j\land i)$ is false, $\lvert V_{c_i}\rvert = x$, and $\dyn[j,y,b-1]\ge 0$, set $\score(i,x,b,j,y) = \dyn[j,y,b-1] + \binom{x}{2}$.
        \item In all other cases, set $\score(i,x,b,j,y) = -1$.
    \end{itemize}

    Note that $\score(i,x,b,j,y)$ calculates the number of pairs that can be obtained by starting with the committee $W[j,y,b-1]$ and then adding $c_i$ so that the connected component of $c_i$ has $x$ voters. We then define $\dyn[i,x,b]$ and $W[i,x,b]$ so as to maximize this quantity.
    \begin{itemize}
        \item Pick $(j^*, y^*)$ from $\arg\max_{(j, y): j\in [i-1], y\in\{0\}\cup[n]} \score(i,x,b,j,y)$.
        \item Set $\dyn[i,x,b] = \score(i,x,b,j^*,y^*)$.
        \item Set $W[i,x,b] = W[j^*,y^*,b-1]\cup \{c_i\}$ if $\dyn[i,x,b]\ge 0$.
    \end{itemize}

    \paragraph{Correctness of the dynamic program.}
    By construction, it holds that if $\dyn[i,x,b]$ is nonnegative, then we have $\Connections(W[i, x, b])=\dyn[i, x, b]$.
    It remains to prove correctness of the update formulas for the dynamic program, i.e., to show that $\dyn[i,x,b] = \opt[i,x,b]$ for all $i\in [m]$, all $x\in\{0\}\cup[n]$ and all $b\in[k]$. To this end, we proceed by induction on $b$, and for each fixed value of $b$, by induction on $i$ (note that we have already argued that our dynamic program is correct for $b=1$ and for $i=1$).
    Thus, we fix a triple $(i,x, b)$ and assume that $\dyn[j, y, b-1]=\opt[j,y, b-1]$ for all $j<i$ and all $y\in [0]\cup\{n\}$; our goal is to show that
    $\dyn[i, x, b]=\opt[i,x, b]$.
    We
    split the proof into two parts.

    First, we will argue that $\dyn[i, x, b]\le\opt[i, x, b]$.
    Clearly, this is true if $\dyn[i, x, b] = -1$.
    Otherwise, $W[i,x,b]=W[j, y, b-1]\cup\{c_i\}$
    for some $j<i$, $y\in\{0\}\cup[n]$ that maximize $\score(i,x,b,j,y)$. Since $\dyn[i, x, b]$ is positive, the quantity $\score(i,x,b,j,y)$ is positive as well. Consequently,
    if $\mathbb{1}(j\land i)$ is false, then necessarily $x=|V_{c_i}|$,
    while if $\mathbb{1}(j\land i)$ is true, then necessarily $x = n(\lnot j,i) + y$.

    By our previous observations, if $\mathbb{1}(j\land i)$ is true,
    the number of pairs interconnected by $W[i,x,b]$ is equal to
    \begin{align*}
      \conpair(W[i,x,b]) &= \conpair(W[j,y, b-1]) + \binom{x-y}{2} +(x-y) y \\
      &=\dyn[j,y,b-1] + \binom{ x-y }{2} +y(x-y)\\
       & = \score(i,x,b,j,y)
       = \dyn[i,x,b]\text,
    \end{align*}
    and if $\mathbb{1}(j\land i)$ is false,
    the number of pairs interconnected by $W[i,x,b]$ is equal to
    \begin{align*}
       \conpair(W[i,x,b]) & = \conpair(W[j,y,b-1]) + \binom{ x}{2}
       = \dyn[j,y,b-1] + \binom{ x}{2}\\
        & = \score(i,x,b,j,y)
        = \dyn[i,x,b]\text;
    \end{align*}
    in both cases, the second transition uses the inductive hypothesis.
    As $W[i,x,b]$ is a committee of size at most $b$ that has $c_i$ as its rightmost candidate, with $c_i$ being in a connected component that contains $x$ voters, we conclude that $\dyn[i,x,b]\le \opt[i,x,b]$.

    Next, we will show that $\dyn[i, x, b]\ge\opt[i, x, b]$.
    First, if there is no committee of size at most $b$ that has $c_i$ as its rightmost candidate, with $c_i$ being in a connected component containing $x$ voters, then $\opt[i, x, b] = -1$, and the inequality is true.
    Otherwise, consider some such committee $W^*$ with $\Connections(W^*)=\opt[i,x,b]$.
    Since $b\ge 2$ and $i\ge 2$, we may assume without loss of generality that $\lvert W\rvert \ge 2$.
    We set $W' = W^*\setminus \{c_i\}$, let $j$ be the rightmost candidate in $W'$, and let $y$ be the size of $j$'s connected component with respect to $W'$.
    If $\mathbb{1}(j\land i)$ is true, we have
    \begin{align*}
        \dyn[j,y,b-1] &=
        \opt[j,y,b-1]
        \ge \conpair(W')
        =   \conpair(W^*) - \binom{x-y}{2} -(x-y) y
        =   \opt[i,x,b] - \binom{x-y}{2} -(x-y) y\text,
    \end{align*}
    and if $\mathbb{1}(j^*\land i)$ is false, we have
    \begin{align*}
        \dyn[j,y,b-1] & =
        \opt[j,y,b-1]
        \ge \conpair(W')
         =  \conpair(W^*) - \binom{x}{2}
         =  \opt[i,x,b] - \binom{x}{2}\text.
    \end{align*}
    In the first case, we have $\dyn[i,x,b] \ge \score(i,x,b,j,y) = \dyn[j,y,b-1] + \binom{x-y}{2} +(x-y) y \ge \opt[i,x,b]$. Similarly, in the second case, we have $\dyn[i,x,b] \ge \score(i,x,b,j,y) = \dyn[j,y,b-1] + \binom{x}{2} \ge \opt[i,x,b]$.

    Together, we obtain $\dyn[i,x,b] = \opt[i,x,b]$.
    Finally, to compute a feasible committee that maximizes \Connections{}, we output an arbitrary committee $W\in\arg\,\max_{i\in [m], x\in \{0\}\cup[n]}\Connections(W[i, x, k])$.

    Note that our dynamic program has $\orderof(mnk)$ cells,
    as $i\in [m]$, $x\in \{0\}\cup[n]$, and $b\in [k]$.
    Moreover, every cell can be filled in polynomial time given the values of previously computed cells.
    Hence, we have obtained a polynomial-time algorithm to compute a committee that maximizes \Connections{}.
\end{proof}

\section{Approximation of EJR via the Method of Equal Shares}\label{app:MES}

We define a parameterized version the method of equal shares, which we will refer to as $\alpha$-MES; the parameter $\alpha \in(0, 1)$ indicates which fraction of the `budget' $k$ is made available to the voters. The standard version of this rule corresponds to $\alpha=1$.

On an instance $(V, A,k)$ with $|V|=n$, under $\alpha$-MES every voter $v\in V$ is assigned a \emph{per-voter budget} $\mathrm{bud}(v) = \alpha\cdot \frac{k}{n}$.
The committee $W$ is initialized as the empty set, and every candidate is assumed to have a cost of $1$.
Informally, at each step, we consider the candidates in $C\setminus W$ that can be afforded by the voters approving them while sharing the costs equally (with the caveat that a voter who cannot afford to pay an equal share can contribute her entire remaining budget instead). We define a candidate's {\em price} $\rho$ as the maximum amount that a supporter of this candidate contributes to its cost.
We then add a candidate with a minimum price to $W$, and proceed to the next step.

More formally, at the start of each step we let
$$
C^* = \{c\in C\setminus W: \sum_{v\colon c\in A_v}\mathrm{bud}(v)\ge 1\}\text.
$$
If $C^*=\varnothing$, the algorithm terminates and returns $W$. Otherwise, for each $c\in C^*$
we set
$$
\rho (c)=\min\{\rho\ge 0: \sum_{v\colon c\in A_v} \min{(\mathrm{bud}(v),\rho)}\ge 1\}\text.
$$
We then pick $c^*$ in $\arg\min_{c\in C^*}\rho(c)$ and add it to $W$.
In addition, we update the budget of each voter $v$ with $c^*\in A_v$ as
$\mathrm{bud}(v) - \min\{\mathrm{bud}(v),\rho(c^*)\}$.

We are ready to prove the EJR guarantee achieved by this method.

\MESalpha*
\begin{proof}
    At the start of the procedure, the sum of voters' budgets is $\alpha k$, and in each
    round the total budget is reduced by $1$. Hence,
    $\alpha$-MES terminates after at most $\lfloor \alpha k \rfloor$ rounds, returning a committee of size at most $\lfloor \alpha k \rfloor$.
    Assume without loss of generality that
    $W=\{c_1, \dots, c_r\}$, where $r\le \lfloor \alpha k \rfloor$ and for each $t\in [r]$ it holds that candidate $c_t$ is added to $W$ in round $t$.

    Assume for contradiction that $W$ violates $\alpha$-\EJR{}.
    Hence, there is a subset of voters $S\subseteq V$ and a positive integer $\ell \le k$ with $\lvert S\rvert \ge \frac{\ell n}{\alpha k}$ and $|\bigcap_{v \in S} A_v| \ge \ell$ such that $\lvert A_v \cap W \rvert < \ell$ for all $v\in S$. Observe that the set $\bigcap_{v\in S} A_v\setminus W$ is non-empty, and
    let $c^*$ be some candidate in this set.

    We claim that whenever a voter $v\in S$ makes a positive contribution towards the cost of some candidate during the execution of the algorithm,
    she pays at most $\frac{\alpha k}{\ell n}$.
    Indeed, suppose that this is not the case,
    and consider the first round~$t$ in which a voter $v\in S$ spends strictly more than $\frac {\alpha k}{\ell n}$.
    By definition of MES, this means that in round $t$ we have $\rho(c_t) > \frac{\alpha k}{\ell n}$. On the other hand, in each round $i<t$, each voter in $S$ spends at most $\frac{\alpha k}{\ell n}$. Moreover, by our assumption, each voter in $S$ approves at most $ \ell-1$ candidates from $W$.
    Thus, before round $t$, the remaining budget of each voter in $S$ is at least $\frac{\alpha k}{n}- (\ell -1) \cdot \frac{\alpha k}{\ell n} = \frac{\alpha k}{\ell n}$. Thus, together they have a budget of at least $|S|\cdot\frac{\alpha k}{n\ell} \ge \frac{\ell n}{\alpha k} \cdot \frac{\alpha k}{\ell n} = 1$. This means that the candidate $c^*\in \bigcap_{v\in S} A_v\setminus W$ belongs to the set $C^*$ at the start of round $t$ and, moreover, $\rho(c^*)\le \frac{\alpha k}{\ell n}<\rho(c_t)$, a contradiction with the choice of~$c_t$.

    Now, since each voter in $S$ contributes at most $\frac{\alpha k}{\ell n}$ towards the cost of each candidate in $W$, and she pays for at most $\ell-1$ candidates in $W$, at the end of round $r$ the remaining budget of each voter in $S$
    is at least $\frac{\alpha k}{\ell n}$, so voters in $S$
    can still afford candidate $c^*$, a contradiction with the assumption that $\alpha$-MES terminates after $r$ rounds.
\end{proof}

\end{document}